\documentclass[journal]{IEEEtran}
\usepackage{array}
\usepackage{epsfig}
\usepackage{cite}
\usepackage{epstopdf}
\usepackage{amsfonts,amsmath,amsthm,amssymb,balance}
\usepackage{multirow,balance}
\usepackage{bm}
\usepackage[usenames,dvipsnames]{color}
\usepackage{colortbl}
\usepackage{float}
\usepackage[ruled,vlined,linesnumbered]{algorithm2e}
\usepackage{graphicx}
\usepackage{thmtools}
\usepackage{subcaption}

\usepackage{url}
\definecolor{mygray}{gray}{.9}

\graphicspath{{figure/}}
\newcolumntype{C}[1]{>{\PreserveBackslash\centering}p{#1}}
\newcolumntype{R}[1]{>{\PreserveBackslash\raggedleft}p{#1}}
\newcolumntype{L}[1]{>{\PreserveBackslash\raggedright}p{#1}}
\newcommand{\AlgoResetCount}{\renewcommand{\@ResetCounterIfNeeded}{\setcounter{AlgoLine}{0}}}
\newcommand{\AlgoNoResetCount}{\renewcommand{\@ResetCounterIfNeeded}{}}
\newcounter{AlgoSavedLineCount}

\theoremstyle{definition}
\newtheorem{theorem}{Theorem}

\newtheorem{lemma}{Lemma}

\def\BibTeX{{\rm B\kern-.05em{\sc i\kern-.025em b}\kern-.08em
    T\kern-.1667em\lower.7ex\hbox{E}\kern-.125emX}}

\begin{document}

\title{Energy-Efficient and Physical Layer Secure Computation Offloading in Blockchain-Empowered Internet of Things}

\author{\IEEEauthorblockN{Yiliang Liu, ~\IEEEmembership{Member,~IEEE}, Zhou Su,~\IEEEmembership{Senior Member,~IEEE}, and Yuntao Wang, ~\IEEEmembership{Student Member,~IEEE}}

\thanks{Y. Liu (email: {\tt liuyiliang@xjtu.edu.cn}), Z. Su (email: {\tt zhousu@ieee.org}), and Y. Wang (email: {\tt yuntao.wang@stu.xjtu.edu.cn}) are with the School of Cyber Science and Engineering, Xi'an Jiaotong University, Xi'an 710049, Shaanxi, China. (Corresponding author: Z. Su.)
}
\thanks{This work was supported in part by National Natural Science Foundation of China (Nos. 62101429, U20A20175, and U1808207).}

}

\maketitle

\begin{abstract}
This paper investigates computation offloading in blockchain-empowered Internet of Things (IoT), where the task data uploading link from sensors to a base station (BS) is protected by intelligent reflecting surface (IRS)-assisted physical layer security (PLS). After receiving task data, the BS allocates computational resources provided by mobile edge computing (MEC) servers to help sensors perform tasks. Existing blockchain-based computation offloading schemes usually focus on network performance improvements, such as energy consumption minimization or latency minimization, and neglect the Gas fee for computation offloading, resulting in the dissatisfaction of high Gas providers. Also, the secrecy rate during the data uploading process can not be measured by a steady value because of the time-varying characteristics of IRS-based wireless channels, thereby computational resources allocation with a secrecy rate measured before data uploading is inappropriate. In this paper, we design a Gas-oriented computation offloading scheme that guarantees a low degree of dissatisfaction of sensors, while reducing energy consumption. Also, we deduce the ergodic secrecy rate of IRS-assisted PLS transmission that can represent the global secrecy performance to allocate computational resources. The simulations show that the proposed scheme has lower energy consumption compared to existing schemes, and ensures that the node paying higher Gas gets stronger computational resources.
\end{abstract}


\section{Introduction}
With the rapid development of wireless communications and sensor manufacturing technologies, the narrowband-enabled Internet of Things (IoT) technologies, released by 3GPP, have gained increasing attention in industrial, medical, business fields, etc \cite{Gozalvez2016}. Following that, the security issues of heterogeneous devices and infrastructures have raised many concerns. Building on their characteristics of decentralization, lightweight, and multimoding, the blockchain and physical layer security are regarded as enabling endogenous security technologies to simultaneously address trust and confidentiality issues of IoT networks \cite{Wang2021a,Liu2021,Su2020}.

The typical application of the blockchain-empowered IoT is the computation offloading of the tasks of sensor nodes. Smart contracts perform the automated and trust-free offloading process between sensor nodes and mobile edge computing (MEC) servers on top of the blockchain to ensure that transactions cannot be denied and malicious behaviors will also be traced \cite{Wu2021, Xie2021, Dai2021}. However, existing IoT computation offloading schemes have the following two challenges. Firstly, traditional schemes usually focus on the optimization of computation or communication performance, such as computing latency minimization, energy consumption minimization, or network throughput maximization \cite{Mu2020,Han2021iot,Liu2021,Dai2018}. The Gas\footnote{Gas refers to the cost necessary to perform transactions on Ethereum, where Gas price has the unit ETH or Gwei. The fee of a transaction can be calculated by the amount of Gas $\times$ Gas price. Gas price has a transformation to legal tenders. For instance, a transaction needs 100 Gwei with $6.5\times 10^{-7}$ dollar/Gwei.} factor only affects the block generation speed and is not considered in computational resource allocation, leading to the dissatisfaction of sensors as better computational resources are not allocated to those sensors under the performance priority strategy even if they pay high Gas. Secondly, the coverage of the central trusted authority (TA) is limited in dynamic and pervasive IoT networks. Cryptographic technologies for the establishment and update of secure communications rely on TA \cite{Jiang2020,Li2021}. Hence, the security of links from the sensors to the MEC servers cannot be protected effectively. In addition, due to the cost constraint of sensor devices, traditional PLS schemes, such as multiple-antenna beamforming or artificial noise technologies \cite{Liu2020}, are not suitable for IoT networks.
 
To address these problems mentioned above, we present an intelligent reflecting surface (IRS)-assisted secure computation offloading scheme to guarantee that the node paying higher Gas has more opportunities to get a stronger computational resource, while reducing the energy consumption of the computation offloading process. The main contributions are summarized as follows.
\begin{itemize}
\item We formulate a PLS-assisted computation offloading model in blockchain-empowered IoT, where phase shift matrix adjustment of IRS devices, multiple antenna technologies at BS, and the allocation of computational resources of MEC servers with dissatisfaction thresholds of sensors are jointly considered to reduce the energy consumption of the computation offloading process.
\item We propose a phase shift matrix adjustment scheme via manifold optimization to improve the secrecy rate. Then, we deduce the expression of the ergodic secrecy rate of this adjustment scheme, which provides a global metric of secrecy performance for computational resource allocation. 
\item We transform the original computational resource allocation problem into a 2-dimensional matching problem of bipartite graphs, which is solved by the Kuhn-Munkres (KM) algorithm. Especially, the edges of the bipartite graph are restricted with dissatisfaction degrees of sensors, so the optimal matching achieves the goal that the sensors with higher Gas payment are prioritized with better computational resources. After computational resource allocation, we propose an adaptive PLS coding method to improve the effective secrecy rate during the data uploading process.
\end{itemize}

The remainder of the paper is organized as follows. Section \ref{related} surveys the related works. Section \ref{model} describes the system model and problem formulation. The IRS-assisted PLS and computation offloading schemes are proposed in Section \ref{proposed1} and \ref{proposed2}, respectively. We show simulation results in Section \ref{simulations}, and conclude this paper in Section \ref{conclusions}.

\textsl{Notations:} Bold uppercase letters, such as $\mathbf{A}$, denote matrices, and bold lowercase letters, such as $\mathbf{a}$, denote column vectors. $\mathbf{A}^{\rm{H}}$, $\mathbf{A}^{\rm{T}}$, and $\mathbf{A}^{\dagger}$ represent the conjugate transpose, transpose, and conjugate transformation of $\mathbf{A}$, respectively. $\mathbf{I}_a$ is an identity matrix with its rank $a$. $\mathcal{CN}(\mu,\sigma^2)$ is a complex normal (Gaussian) distribution with mean $\mu$ and variance $\sigma^2$. $|\mathbf{x}|$ is the Euclidean norm of $\mathbf{x}$. $\text{diag}(\mathbf{X})$ is the diagonal matrix of $\mathbf{X}$. $\mathbb{E}(\cdot)$ is the expectation operation. $\text{arg}(x)$ is the angle of complex variable $x$. $\circ$ is the Hadamard product. $\mathfrak{R}(x)$ means the real part of complex variable $x$.

\section{Related Work}\label{related}

\subsection{Blockchain-empowered IoT}
The blockchain technology is now integrated into the IoT networks to ensure trust and traceability of node behaviors. Despite the benefit that blockchain brings to traditional IoT applications, there are still many challenges in its actual implementation, such as blockchain data storage, power consumption, and task offloading. To cope with these issues, many efforts to adopt blockchain in IoT applications have been conducted in recent years.

\subsubsection{Data Storage and Prototype}
Yu \emph{et al.} presented a multiple-layer storage architecture for large-scale IoT networks to meet the storage requirement of massive transaction data of blockchain \cite{Yu2021}. Furthermore, Pyoung \emph{et al.} proposed finite lifetime blocks to reduce storage costs in IoT, where outdated transactions and blocks can be safely removed from the blockchain \cite{Pyoung2020}. Pan \emph{et al.} designed the prototype of an edge IoT resource management system based on blockchain and smart contracts, where all the IoT activities and transactions are recorded on the blockchain for security purposes \cite{Pan2019}.

\subsubsection{Computation Offloading}
In this case,  Xu \emph{et al.} safeguarded data integrity during real-time data processing in blockchain-empowered IoT \cite{Xublock2020}, while considering the imbalances of computational resource distribution and workload. Due to the computation and storage constraints of mobile IoT devices in blockchain-enabled applications, Liu \emph{et al.} offloaded compute-intensive proof-of-work (PoW) mining missions to nearby edge servers \cite{Liublock2018}. Chen \emph{et al.} further investigated joint computation offloading of mining tasks and data processing tasks for industrial IoT (IIoT) devices and used Nash equilibrium of the game theory to minimize their economic costs \cite{Chenblock2019}. Feng and Nguyen \emph{et al.} designed a novel deep reinforcement learning algorithm to improve the communication and computation performance during the entire computation offloading process \cite{Fengblock2020,Nguyen2021}. Considering the cost of resources and income of the service provisioning, Fan \emph{et al.} used the Stackelberg game where computation resource providers are set as the leader, and IoT users are followers~\cite{Fan2021}. The objective of Fan's scheme is to optimize the revenues of computation resources. Seng \emph{et al.} proposed a matching algorithm between IoT users and edge servers by considering task execution time and energy consumption, and the matching algorithm is implemented by smart contracts without relying on central nodes \cite{Seng2021}. Our work focuses on the confidentiality of the task data uploading process via PLS schemes and the Gas effect on the computational resource allocation, which are not discussed in previous researches. 

\subsection{IRS-assisted PLS Technologies}
IRS devices are usually applied to satellite communications and radars but are not widely used in the mobile communication field because of cost and fabrication constraints. Thanks to the development of new material technologies, IRS has gradually been used in mobile communications and networks to improve signal quality and security performance. In the area of IRS-assisted PLS, many investigations are conducted to improve the secrecy rate, such as the semidefinite relaxation (SDR)-based and majorization minimization (MM)-based phase shift matrix optimization methods \cite{Sai2021,Zheng2021}. However, there are very few researches on expressions of security performance metrics. Yang and Trigui \emph{et al.} considered the IRS-assisted single-input single-output single-antenna eavesdropper (SISOSE) case, and gave its expressions of ergodic secrecy rate and secrecy outage probability \cite{Yang2020,Trigui2021}. For multiple-input single-output multiple-antenna eavesdropper (MISOME) scenarios, the Monte Carlo method is used to obtain the ergodic secrecy rate with a lot of computational overhead \cite{FengIRS2020}. Our paper considers the uplink communications of IoT with a passive multiple-antenna eavesdropper, which obeys IRS-assisted single-input multiple-output multiple-antenna eavesdropper (SIMOME), and its ergodic secrecy rate expression is not investigated in existing works.

\subsection{Discussion}

The Gas factor is not considered in existing investigations of computation offloading in blockchain-empowered IoT networks. Besides, the absence of expressions of ergodic secrecy rate and effective secrecy rate of IRS-assisted SIMOME results in that computational resource allocation and data uploading lack optimization objectives and parameters. In this work, we first design a manifold-based phase shift matrix adjustment scheme for IRS-assisted PLS. Then, we provide its expressions of ergodic secrecy rate and effective secrecy rate, followed by a Gas-oriented computational resource allocation scheme that can reduce the energy consumption, while ensuring the fairness among sensors.

\begin{figure}[!htp]
\centering
\includegraphics[width=0.97\linewidth]{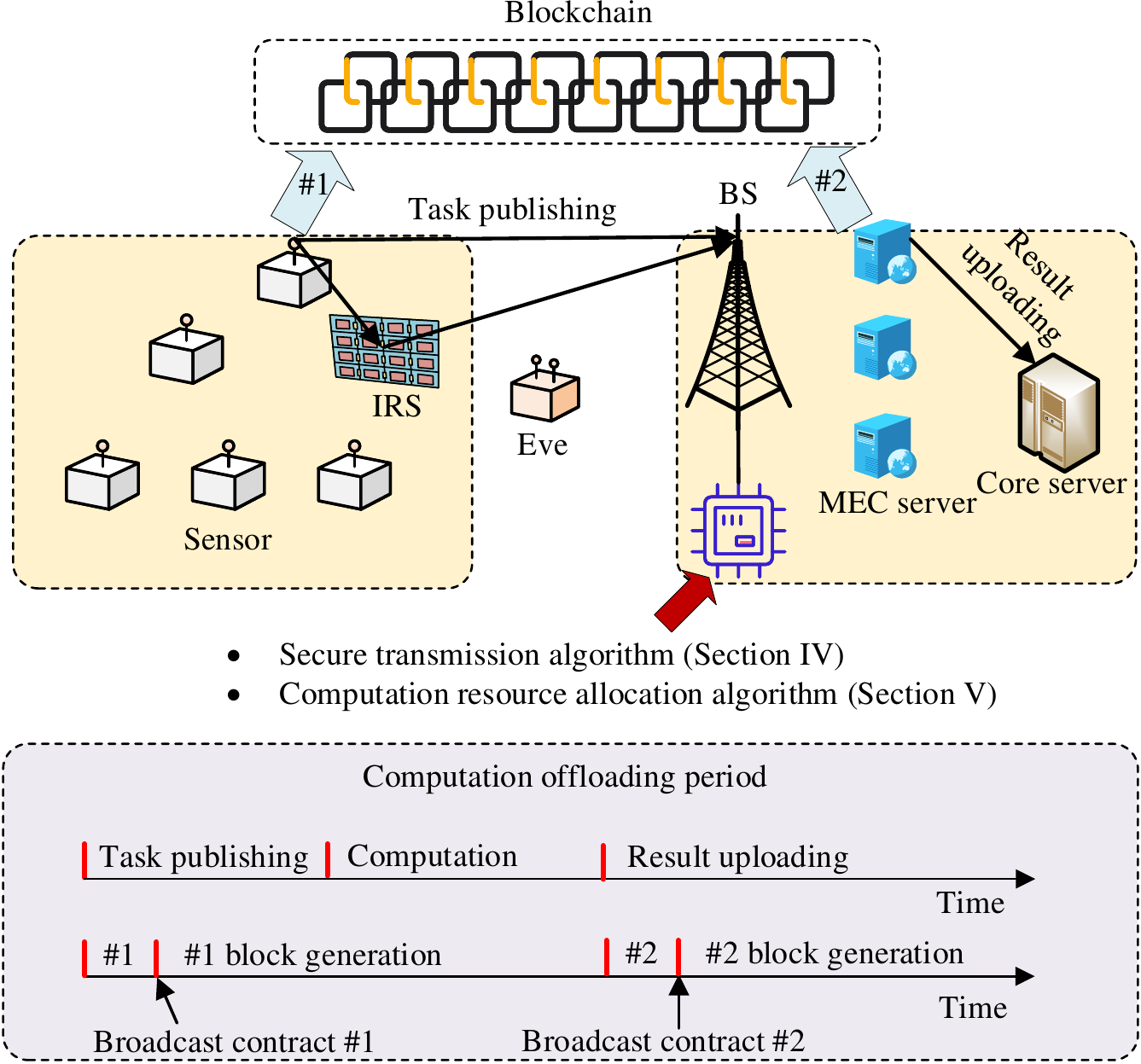}
\caption{Secure computation offloading in IoT networks, where the sensor $U_i$ sends computation tasks to the BS via IRS and instantaneously broadcasts the task publishing contracts (denoted by $\#1$). After receiving tasks, the BS allocates computational resources of MEC server $M_k$ to complete the tasks, then $M_k$ broadcasts the result uploading contracts (denoted by $\#2$). Both task publishing and result uploading should be recorded in blocks via the $\#1$ block generation and $\#2$ block generation, respectively. The uplink links from sensors to the BS is protected by an IRS-assisted PLS scheme. }\label{model_figure}
\end{figure}

\section{System Model and Problem Formulation}\label{model}
This article considers a blockchain-empowered IoT network, as shown in Fig. \ref{model_figure}, which includes $N_I$ sensor nodes, denoted by $\{U_1,U_2,...,U_{N_I}\}$. Each sensor is equipped with one antenna. A set of $N_K$ MEC servers, denoted by $\{M_1,M_2,...,M_{N_K}\}$, are associated with an $N_b$-antenna BS via wired links, and $N_K \geq N_I$. Due to the transmission power constraints and the large-scale fading effect of sensor communications, an IRS device is deployed to enhance the security of the uplink channel from sensors to the BS.

All sensors and MEC servers are registered in Ethereum and follow the rules of Ethereum. Ethereum is a decentralized, open-source, and public blockchain platform supporting smart contracts, including three functional layers \cite{Wu2021arch} \cite{Wu2021netw}.
\begin{enumerate}
\item Underlaying layer: It mainly involves the transmission and storage mechanisms of transaction and block data, as well as the blockchain-based cryptographic algorithms, such as hash and certificates. 
\item Core layer: It includes incentive, consensus, and blockchain topology mechanisms. The incentive mechanism is based on Ether, which introduces economic incentives to make the nodes give their efforts to verify data in Ethereum. The consensus mechanism achieves information synchronization among untrusted parties in decentralized systems, such as Proof of Work (PoW), Proof of Stake (PoS). The blockchain topology mechanism involves transactions, mining, block verification, Gas, etc., where Gas is the cost necessary to perform transactions on Ethereum.
\item Application layer: It involves script codes, smart contracts, and other programmable codes that can be used to enable diverse Ethereum transactions.
\end{enumerate}

\subsection{Ethereum-Assisted Computation Offloading Architecture}
As the computational resources of sensors are scarce, these $N_I$ sensors offload their computational tasks, denoted by $\{Z_1,Z_2,...,Z_{N_I}\}$, to MEC servers via the wireless channels between sensors and the BS. After receiving tasks, the BS allocates virtual machines provided by MEC servers to execute these tasks. To record these tasks, sensors store the indices of publishing tasks on Ethereum via a contract function, i.e., $\mathsf{task\_publish\_contract}(\cdot)$, and the BS stores the indices of results on Ethereum via the contract function $\mathsf{result\_upload\_contract}(\cdot)$. Once the blocks of these contracts are deployed on Ethereum, no one can repudiate the debts of transactions, and the results given by MEC servers are recorded to defend against malicious MEC servers. Our work adopts the Ethereum-assisted computation offloading architecture presented in \cite{Liugit1}, which is described as follows. 
\begin{enumerate}
\item Transmitting tasks to the BS: The sensor $U_i$ sends the task data to the BS via IRS-assisted PLS schemes to resist eavesdropping attacks. To avoid inter-user interference, sensors use the time-division multiple access (TDMA) technology during transmission processes.
\item Creating contracts of publishing tasks: At the same time of task transmission, by following the contract format of Ethereum, the sensor $U_i$ launches a contract, i.e., $\#1=\mathsf{task\_publish\_contract}(Z_i)$, which includes the hash message corresponding to the offloaded task $Z_i$ and the signature of $U_i$. Also, the contract has the address information of the transaction parties\footnote{The sending address is related to $U_i$, and the received address is the constructor of the contract as the contract should be broadcast in Ethereum.}, Gas for this transaction defined by $V_i$, Gas price, etc. Then, the contract $\#1$ is broadcast in Ethereum.
\item Computing tasks in MEC server: After receiving tasks, the BS allocates computational resources of MEC servers to perform these tasks. To avoid the complexity caused by tasks and computing power partition, it is assumed that a task uses one MEC server to compute results and an MEC server can be allocated to only one task. At last, the results are uploaded to the core server for subsequent services.
\item Generating blocks of contracts of task publishing: All nodes in Ethereum synchronize the transaction of $\mathsf{task\_publish\_contract}(Z_i)$, and check its format and signature. If passing, Ethereum nodes compete with each other to win the right of charging the account of the transaction, then the block is generated in  Ethereum by the winners. All members registered in Ethereum desiring to get payments can take part in the competition of account-charging rights.
\item Creating contracts of result uploading: Without loss of generality, we assume that the task of $U_i$, i.e, $Z_i$, is offloaded to the MEC server $M_k$. After completing the task, $M_k$ launches a contract, i.e., $\mathsf{result\_upload\_contract}(Z_i)$, which includes the hash message of  results, the signature of $M_k$, allocated computational resources, address information, Gas for this transaction, Gas price, etc.
\item Generating blocks of contracts of result uploading: Similar to the block generation of contracts of task publishing, all members registered in Ethereum can take part in this competition of account-charging right. The winner can obtain Gas as a payment.
\end{enumerate}

In this system, the IoT service acquirer should pay sensors and MEC servers as the original data is generated by sensors and the data process is done by MEC servers. In this case, an IoT-enabled business model is established. 

\subsection{Communication Model}
The article considers an IRS-assisted uplink communication model with a passive $N_e$-antenna eavesdropper (Eve), an $N_b$-antenna BS, $N_I$ single-antenna sensors, and an IRS device. The IRS device is equipped with $N$ programmable phase shift elements. All channels are assumed to obey Rayleigh fading, i.e., the channel from $U_i$ to Eve is defined as $\mathbf{g}_{i}\sim\mathcal{CN}_{N_e,1}(\mathbf{0},\mathbf{I}_{N_e})$, the wiretap channel aimed at $U_i$ from IRS to Eve is defined as $\mathbf{Z}_{i}\sim\mathcal{CN}_{N_e,N}(\mathbf{0},\mathbf{I}_{N_e}\otimes\mathbf{I}_N)$, the direct link from $U_i$ to BS is defined as $\mathbf{l}_{i}\sim\mathcal{CN}_{N_b,1}(\mathbf{0},\mathbf{I}_{N_b})$, the channel from $U_i$ to IRS is defined as $\mathbf{h}_{i}\sim\mathcal{CN}_{N,1}(\mathbf{0},\mathbf{I}_N)$, and the channel from IRS to BS for the transmission of $U_i$ is defined as $\mathbf{A}_{i}\sim\mathcal{CN}_{N_b,N}(\mathbf{0},\mathbf{I}_{N_b}\otimes\mathbf{I}_N)$. The instantaneous CSIs of legitimate devices, including $\mathbf{l}_i$, $\mathbf{h}_i$, and $\mathbf{A}_i$, can be perfectly obtained via channel estimation, whereas the instantaneous Eve's CSIs $\mathbf{g}_{i}$ and $\mathbf{Z}_{i}$ are unknown. 

For the security purpose, IRS controls programmable phase shift elements via a phase shift matrix, where the phase shift matrix for the transmission period of $U_i$ is defined as an $N\times N$ matrix $\bm{\Phi}_i$, i.e., 
\begin{flalign}\label{psm}
\bm{\Phi}_i=\text{diag}[\exp(j\theta_{i,1}),...,\exp(j\theta_{i,n}),..., \exp(j\theta_{i,N})],
\end{flalign}
where $j=\sqrt{-1}$, and $\theta_{i,n}\in [0,2\pi)$ is the phase introduced by the $n$-th phase shift element of IRS at the $i$-th period. With the phase shift matrix $\bm{\Phi}_i$, the received signals at the BS and Eve can be expressed as
\begin{flalign}
& \mathbf{y}_i = \alpha_i(\mathbf{l}_{i}+\mathbf{A}_i\bm{\Phi}_i\mathbf{h}_i)x_i+\mathbf{n}_i, \label{mchannel}\\
& \mathbf{y}_{e,i} =\alpha_{e,i}(\mathbf{g}_{i}+\mathbf{Z}_{i}\bm{\Phi}_i\mathbf{h}_i)x_i+\mathbf{n}_{e,i}, \label{Le2}
\end{flalign}
where $\alpha_i$ is the path loss between the BS and $U_i$, $\alpha_{e,i}$ is the path loss between the BS and Eve, $x_i$ is the confidential information-bearing signal from $U_i$ with $\mathbb{E}(|x_i|^2)=P_i$, and $P_i$ is the transmission power of $U_i$. $\mathbf{n}_i$ and $\mathbf{n}_{e,i}$ are the additive white Gaussian noise (AWGN) obeying $\mathcal{CN}_{N_b,1}(0,\sigma^2_i\mathbf{I}_{N_b})$ and $\mathcal{CN}_{N_e,1}(0,\sigma_{e,i}^2\mathbf{I}_{N_e})$, respectively. 

\subsection{Energy Consumption Model}

The sensor $U_i$ has the computation task $Z_i$ with $D_i$-bits data that should be uploaded to the BS. After receiving the task data, the BS will select an MEC server, such as $M_k$, to perform the task of $U_i$. In this model, when $D_i$-bits data should be calculated in $M_k$, the computing period is $c_kD_i/f_k$, where $c_k$ (CCN/bit) is the CPU cycle number (CCN) per bit processing at $M_k$, representing computing efficiency of CPU chips, and $f_k$ (CCN/s) is the CCN per second of $M_k$. The energy consumption per second of $M_k$ is $\eta_k f_k^3$ (Joule/s) \cite{Dai2018}, where $\eta_k$ is the computation energy efficiency coefficient of CPU chips in $M_k$ depending on the chip architecture. The transmission delay of uploading can be expressed as $D_i/(R_{i}B)$, where $R_i$ (bit/s/Hz) is the secrecy rate of uplink channels, and $B$ is the bandwidth. The energy consumption of $Z_i$ with the assistance of $M_k$ can be formulated as
\begin{flalign}\label{ec} 
Q_{i,k} = \eta_k c_k D_i  f_k^2+\frac{D_i P_i}{R_iB}.
\end{flalign}
Note that the computing modules of Ethereum are worldwide distribution, and are independent of the computation offloading system, so the energy consumption of Ethereum is not considered here.

\subsection{Problem Formulation}
In traditional Ethereum systems, Gas payment affects the time and energy consumption of block generation, but has no effect on computational resource allocation. In real scenarios, sensors paying higher Gas have more desire to get better computational resources, rather than recording their behaviors on Ethereum. Hence, when $U_i$ is allocated $M_k$, we define the degree of dissatisfaction of $U_i$ as follows,
\begin{flalign}\label{unsatisfactory}
O_i=W_{i,k}\bigg|r(V_i)-r\bigg(\frac{f_k}{c_k}\bigg)\bigg|,
\end{flalign}
where $W_{i,k}$ is the element of the $i$-th row and $k$-th column of $\mathbf{W}$. $W_{i,k}$ is a binary variable taking 1 when $M_k$ is assigned to $U_i$, and 0 otherwise. $r(V_i)$ is the index of the descending order of Gas set provided by $N_I$ sensors, and $r(f_k/c_k)$ is the index of the descending order of computing power of $N_K$ MEC servers.  

With the consideration of the energy consumption and degrees of dissatisfaction, the problem of the computation offloading focuses on the energy consumption minimization with constraints as follows.
\begin{flalign}
\text{P1: }& \min_{\bm{\Phi}_i,\forall i, \mathbf{W}}\sum_{i=1}^{N_I}\sum_{k=1}^{N_K}W_{i,k}Q_{i,k},\\
& \text{s.t. } O_{i}\leq \epsilon, \forall i, \label{p1c1} \\ 
&\quad \; W_{i,k}= \{0 \text{ or }1\}, \forall i,k, \label{p1c2} \\
&\quad \; \sum_{i=1}^{N_I}W_{i,k}\leq 1, \quad \sum_{k=1}^{N_K}W_{i,k}\leq 1, \label{p1c3} \\
&\quad \; \text{Eq. }(\ref{psm}), \label{p1c4}
\end{flalign}
where Eq. (\ref{p1c1}) requires that the degrees of dissatisfaction should be small than a threshold $\epsilon$ for all sensors, and $\epsilon \in [0,1,...,N_I-1]$ can be adjusted manually. Eq. (\ref{p1c2}) is the binary constraint representing the allocation factor. Eq. (\ref{p1c3}) reveals that each sensor can be allocated with only one MEC server, and each MEC server is only assigned to one sensor. Eq. (\ref{p1c4}) is the passive phase shift matrix constraint.

It is obvious that $\bm{\Phi}_i$ and $\mathbf{W}$ are coupled variables in the optimization of the $i$-th period, and P1 is non-convex mixed-integer problem. To tackle the problem, we transform P1 into two sub-problems, i.e., phase shift optimization and computational resource allocation, then solve them step-by-step. 

\section{Phase Shift Optimization for PLS and Its Ergodic Secrecy Rate}\label{proposed1}

From Eq. (\ref{ec}), we find that the energy consumption $Q_{i,k}$ decreases with the increasing secrecy rate $R_i$ for all sensors. To reduce the energy consumption in the data uploading process, we optimize $\bm{\Phi}_i$ to maximize $R_i, \forall i$.

\subsection{Phase Shift Optimization}
According to Eqs. (\ref{mchannel}) and (\ref{Le2}), the secrecy rate $R_i$ between $U_i$ and BS is given as follows \cite{Liu2020},
\begin{flalign}\label{esr}
R_i& = (C_{m,i}-C_{w,i})^+,
\end{flalign}
where
\begin{flalign}
&C_{m,i}=\log_2\bigg(1+\frac{\alpha_i^2P_i}{\sigma_i^2}|\mathbf{l}_{i}+\mathbf{A}_i\bm{\Phi}_i\mathbf{h}_i|^2\bigg), \label{cm} \\ 
&C_{w,i}= \log_2\bigg(1+\frac{\alpha_{e,i}^2P_i}{\sigma_{e,i}^2}|\mathbf{g}_{i}+\mathbf{Z}_{i}\bm{\Phi}_i\mathbf{h}_i|^2\bigg). \label{cw}
\end{flalign}
The objective is to find the optimal $\bm{\Phi}_i$ to achieve $\{R_i^*,\bm{\Phi}^*\}=\max_{\bm{\Phi}}R_i$. However, it is hard to maximize $R_i$ as $\mathbf{g}_{i}$ and $\mathbf{Z}_i$ are unknown, so the objective is transformed to maximize the channel gain between $U_i$ and the BS as follows, 
\begin{flalign}
\text{P2: }& \max_{\bm{\Phi}} |\mathbf{l}_{i}+\mathbf{A}_i\bm{\Phi}_i\mathbf{h}_i|^2, \\
& \text{s.t. } \text{Eq. }(\ref{psm}). \label{p2c1}
\end{flalign}
P2 is a classic Boolean quadratic program problem. The traditional methods to solve P2 are the SDR technique \cite{Wu2019} and manifold optimization \cite{Yuiccc2019}. The objective parameter $\bm{\Phi}_i$ is the middle of two CSI parameters and is hard to handle. Hence, refer to \cite{Yuiccc2019}, $\mathbf{A}_i\bm{\Phi}_i\mathbf{h}_i$ can be re-formulated as 
\begin{flalign}\label{transm}
\mathbf{A}_i\bm{\Phi}_i\mathbf{h}_i=\mathbf{A}_i\text{diag}(\mathbf{h}_i)\mathbf{q}_i=\bm{\Theta}_i\mathbf{q}_i,
\end{flalign}
where $\mathbf{q}_i=\text{vec}(\bm{\Phi}_i)=[\exp(j\theta_{i,1}),..., \exp(j\theta_{i,N})]^{\rm{T}}$. With Eq. (\ref{transm}), the objective of P2 can be transformed as follows,
\begin{flalign}\label{p2}
& |\mathbf{l}_i+\mathbf{A}_i\bm{\Phi}_i\mathbf{h}_i|^2 \notag \\
&=\mathbf{q}_i^{\rm{H}}\bm{\Theta}_i^{\rm{H}}\bm{\Theta}_i\mathbf{q}_i+|\mathbf{l}_i|^2+\mathbf{l}_i^{\rm{H}}\bm{\Theta}_i\mathbf{q}_i+\mathbf{q}_i^{\rm{H}}\bm{\Theta}_i^{\rm{H}}\mathbf{l}_i.
\end{flalign}
Then, the minus of the objective in P2 is expressed as a function with respect to $\mathbf{q}_i$, i.e.,
\begin{flalign}\label{mo1}
 f(\mathbf{q}_i)&=-(\mathbf{q}_i^{\rm{H}}\bm{\Theta}_i^{\rm{H}}\bm{\Theta}_i\mathbf{q}_i+|\mathbf{l}_i|^2+\mathbf{l}_i^{\rm{H}}\bm{\Theta}_i\mathbf{q}_i+\mathbf{q}_i^{\rm{H}}\bm{\Theta}_i^{\rm{H}}\mathbf{l}_i).
\end{flalign}
The constraint (\ref{p2c1}) is defined by a complex circle manifold~\cite{Yuiccc2019}, which is given by
\begin{flalign}\label{oman}
\mathcal{O}=\big\{\mathbf{q}_i\in\mathbb{C}^{N}\big||[\mathbf{q}_i]_n|^2=1,n=1,...,N\big\},
\end{flalign}
where $[\mathbf{q}_i]_n$ is the $n$-th element of $\mathbf{q}_i$. It is obvious that Eq.~(\ref{oman}) is equivalent to constraint (\ref{p2c1}), and the optimal point of original problem P2, i.e., $\mathbf{q}_i^*$ can be found on the complex circle manifold $\mathcal{O}$. Hence, P2 is equivalently transformed as
\begin{flalign}
\text{P3: } &\min_{\mathbf{q}_i\in\mathcal{O}}f(\mathbf{q}_i). \label{P9O} 
\end{flalign}
P3 is regarded as a complex circle manifold optimization problem. To find $\mathbf{q}_i^*$, we use the Riemannian gradient descent algorithm presented in \cite{Yuiccc2019}, where the Riemannian gradient of a manifold point is decided jointly by the Euclidean gradient and the tangent space on manifold $\mathcal{O}$ at this point. Without loss of generality, the $s$-th point (the $s$-th iteration in gradient descent) of manifold $\mathcal{O}$ is $\mathbf{q}_i(s)\in \mathcal{O}$, and the tangent space for $\mathcal{O}$ at the $s$-th point can be expressed as
\begin{flalign}\label{ts}
T_{\mathbf{q}_i(s)}\mathcal{O}=\big\{\mathbf{v}\in\mathbb{C}^{N}:\mathfrak{R}[\mathbf{v}\circ \mathbf{q}_i(s)^{\dagger}]=\mathbf{0}\big\},
\end{flalign}
where $\mathbf{v}$ is the tangent vector at $\mathbf{q}_i(s)$, $\circ$ is the Hadamard product, $\mathfrak{R}(\cdot)$ means the real part of a complex variable, and $(\cdot)^{\dagger}$ is the conjugate operation. The objective function $f(\mathbf{q}_i)$ with respect to $\mathbf{q}_i$ is defined in Eq. (\ref{mo1}), and the Euclidean gradient at $\mathbf{q}_i(s)$ is 
\begin{flalign}
\triangledown_{\mathbf{q}_i(s)}f(\mathbf{q}_i)=-2\bm{\Theta}_i^{\rm{H}}\bm{\Theta_i}\mathbf{q}_i(s)-2 \bm{\Theta}_i^{\rm{H}}\mathbf{l}_i.
\end{flalign}
Among all tangent vectors on $T_{\mathbf{q}_i(s)}\mathcal{O}$, the one that yields the fastest increase of the objective function is defined as the Riemannian gradient \cite{Yuiccc2019}, i.e., $\text{grad}_{\mathbf{q}_i(s)}f(\mathbf{q}_i)$, which is the projection from the Euclidean gradient to the tangent space $\mathcal{O}$ as follows,
\begin{flalign}\label{grad}
 \text{grad}_{\mathbf{q}_i(s)}f(\mathbf{q}_i)&=\triangledown_{\mathbf{q}_i(s)}f(\mathbf{q}_i) \notag \\ 
&-\mathfrak{R}(\triangledown_{\mathbf{q}_i(s)}f(\mathbf{q}_i)\circ \mathbf{q}_i(s)^{\dagger})\circ\mathbf{q}_i(s).
\end{flalign}
In the Riemannian gradient descent process, the next point of $\mathbf{q}_i(s)$ is with the direction of $\varpi_s\mathbf{p}_s$, where $\varpi_s$ is the step size and $\mathbf{p}_s$ is the descent direction vector. However, the movement can not guarantee that the next point $\mathbf{q}_i(s+1)$ is on manifold $\mathcal{O}$. Hence, we introduce the retraction function to map a vector on $T_{\mathbf{q}_i(s)}\mathcal{O}$ onto manifold $\mathcal{O}$, which is given as
\begin{flalign}\label{next}
\mathbf{q}_i(s+1)=\text{R}_{\mathbf{q}_i(s)}(\varpi_s\mathbf{p}_s),
\end{flalign}
where a typical retraction is the normalization function, i.e., $\text{R}_{\mathbf{x}}(\mathbf{y})=\frac{y_i}{|y_i|},\forall i$. Based on the Riemannian gradient and retraction function, we use the conjugate-gradient descent method to find the optimal phase shift matrix, which is shown in Algorithm \ref{psmo}. In conjugate-gradient descent algorithm, the update rule for the search direction on manifolds is given by
\begin{flalign}\label{sd}
\mathbf{p}_{s+1} & =-\text{grad}_{\mathbf{q}_i(s+1)}f(\mathbf{q}_i)+\varphi_{s}\mathcal{T}_{\mathbf{q}_i(s)\to\mathbf{q}_i(s+1)}(\mathbf{p}_{s}),
\end{flalign}
where $\varphi_{s}$ is the Polak-Ribiere parameter, and $\mathcal{T}_{\mathbf{q}_i(s)\to\mathbf{q}_i(s+1)}(\mathbf{p}_{s})$ is the mapping function of the tangent vector $\mathbf{p}_{s}$ from the tangent space $T_{\mathbf{q}_i(s)}\mathcal{O}$ to the tangent space $T_{\mathbf{q}_i(s+1)}\mathcal{O}$. The mapping function is given as
\begin{flalign}
&\mathcal{T}_{\mathbf{q}_i(s)\to\mathbf{q}_i(s+1)}(\mathbf{p}_{s}) \notag \\
&=T_{\mathbf{q}_i(s)}\mathcal{O}\to T_{\mathbf{q}_i(s+1)}\mathcal{O} \notag \\
&=\mathbf{p}_{s}-\mathfrak{R}(\mathbf{p}_{s}\circ[\mathbf{q}_i(s+1)]^{\dagger})\circ\mathbf{q}_i(s+1).
\end{flalign}
Since the second derivative of $f(\mathbf{q}_i)$, i.e., $\triangledown^2_{\mathbf{q}_i}f(\mathbf{q}_i)=-2\bm{\Theta}_i^{\rm{H}}\bm{\Theta}_i$ and $\bm{\Theta}_i^{\rm{H}}\bm{\Theta}_i$ is positive semidefinite, $f(\mathbf{q}_i)$ is concave, such that Algorithm \ref{psmo} based on conjugate-gradient descent cannot be guaranteed to converge to the optimal point~\cite{manopt}. Hence, the output $\bm{\Phi}^*$ in Algorithm \ref{psmo} is the local optimal result.

\begin{algorithm}[t]
\small
\KwData{$\{\mathbf{l}_i, \mathbf{A}_i, \mathbf{h}_i, \forall i\}, N, N_I$}
\KwResult{$\bm{\Phi}_i^*, \forall i$}
\For{$i=1:1:N_I$}{
Initialize $\varpi_0$ and $\varphi_0$\;
Initialize point $\mathbf{q}_i(0)$, get $\text{grad}_{\mathbf{q}_i(0)}f(\mathbf{q}_i)$\;
\While{$|\rm{grad}_{\mathbf{q}_i(s)}f(\mathbf{q}_i)|\leq \psi $}{
Get Riemannian gradient $\text{grad}_{\mathbf{q}_i(s)}f(\mathbf{q}_i)$ by Eq. (\ref{grad})\;
Get conjugate search direction $\mathbf{p}_{s}$ by Eq. (\ref{sd})\;
Get next point $\mathbf{q}_i(s+1)$ by Eq. (\ref{next}) \;
Determine step size $\varpi_s$ and $\varphi_s$ proposed in \cite{absil2009optimization} \;
}
Get $\mathbf{q}_i^*=\mathbf{q}_i(s)$\;
$\bm{\Phi}_i^*=\text{diag}(\mathbf{q}_i^*)$\;
}
\textbf{Procedure End}
\caption{Phase Shift Matrix Optimization based on Conjugate-gradient Descent.}\label{psmo}
\end{algorithm}

The computational complexity of Algorithm \ref{psmo} is analyzed as follows. To meet the convergence condition, the conjugate-gradient descent algorithm requires $N^2$ iterations \cite{shewchuk1994introduction}. Each iteration needs the calculations of Eqs. (\ref{grad}), (\ref{next}), (\ref{sd}), and the parameter update, i.e.,  step 8, which requires $2N^2+4N$, $N$, $N$, and $4N^2$ inner-iterations, respectively. Step 10 needs $N$ iterations to generate $\bm{\Phi}^*$. For $N_I$ sensors, the number of total iterations of Algorithm \ref{psmo} is $N_I(4N^4+6N^3+N)$.

\subsection{Ergodic Secrecy Rate}

The secrecy rate can be viewed as the inherent property for a given instantaneous CSI realization. The ergodic secrecy rate is the average value of secrecy rates among all realizations of CSIs, which is a usual metric for PLS schemes \cite{Liu2020, Liu2021}. As the CSIs of wiretap channels are unknown, to calculate the ergodic secrecy rate, the gain of wiretap channels in Eq. (\ref{Le2}), i.e., $|\mathbf{g}_{i}+\mathbf{Z}_{i}\bm{\Phi}_i\mathbf{h}_i|^2$, can be regarded as a random variable $X$. The Gamma distribution can deduce the probability density function (PDF) of $X$, as shown in Lemma 1. With the PDF of $X$, the expression of ergodic secrecy rate can be derived as security performance metrics for computational resource allocation. The Gamma distribution can represent the sum of multiple independent exponentially distributed random variables, so it is suitable for representing the gain of IRS-assisted channels, such as \cite{VanChien2021}. 

\vspace{0.1 in}

\begin{lemma}
For an $N\times N$ fixed $\bm{\Phi}$ generated by Eq. (\ref{psm}), three independent random variables $\mathbf{g}\sim\mathcal{CN}_{1,N_e}(\mathbf{0},\mathbf{I}_{N_e})$, $\mathbf{h}\sim\mathcal{CN}_{N,1}(\mathbf{0},\mathbf{I}_{N})$, and $\mathbf{Z}\sim \mathcal{CN}_{N_e,N}(\mathbf{0},\mathbf{I}_{N_e}\otimes \mathbf{I}_{N})$, if we have a random variable $x$ as
\begin{flalign}\label{aux1}
&x=|\mathbf{g}+\mathbf{Z}\bm{\Phi}\mathbf{h}|^2, 
\end{flalign}
the PDF and cumulative distribution function (CDF) of $X\sim X(N_e,N,\bm{\Phi})$ can be expressed as
\begin{flalign}
&f_{X}(x)=\frac{1}{\Gamma(\mu)\nu^{\mu}}x^{\mu-1}\exp(-x/\nu), \label{pdf}\\
& F_{X}(x)=1-\frac{1}{\Gamma(\mu)}\Gamma(\mu,\frac{x}{\nu}), \label{cdf}\\
&\mu=\frac{N_e(1+N)^2}{(1+N)^2+(1+N_e)N}, \label{mu} \\
& \nu=1+N+\frac{(1+N_e)N}{1+N},\label{nu} 
\end{flalign}
where $\Gamma(\cdot)$ is the Gamma function, and $\gamma(\epsilon, \eta)$ is the incomplete gamma function defined as follows,
\begin{flalign}\label{ligamma}
\gamma(\epsilon,\eta)=\int_0^{\eta}\exp(-z)z^{\epsilon-1}\text{d}z.
\end{flalign}
\end{lemma}
\begin{IEEEproof}
See Appendix A. 
\end{IEEEproof}

\vspace{0.2 in}

In order to examine the fitting performance of Gamma distributions, we use the KS  method to find the maximum divergences between the actual distribution and the Gamma distribution \cite{weisstein2002crc}, where statistic $D_p$ is defined as the maximum divergences of PDF tests as follows,
\begin{flalign}
&D_p=\max|f_{X}(x)-\hat{f}_{X}(x)|,
\end{flalign}
where $\hat{f}_{X}(x)$ is the actual PDF of $X\sim X(N_e,N,\bm{\Phi})$ obtained by Monte Carlo simulations. As shown in Fig. \ref{KS}, $D_p$ is with an order of magnitude around $10^{-4}$, so the Gamma distributions are in a good agreement with the actual distributions. 

\begin{figure}[t]
\centering
\includegraphics[width=0.95\linewidth]{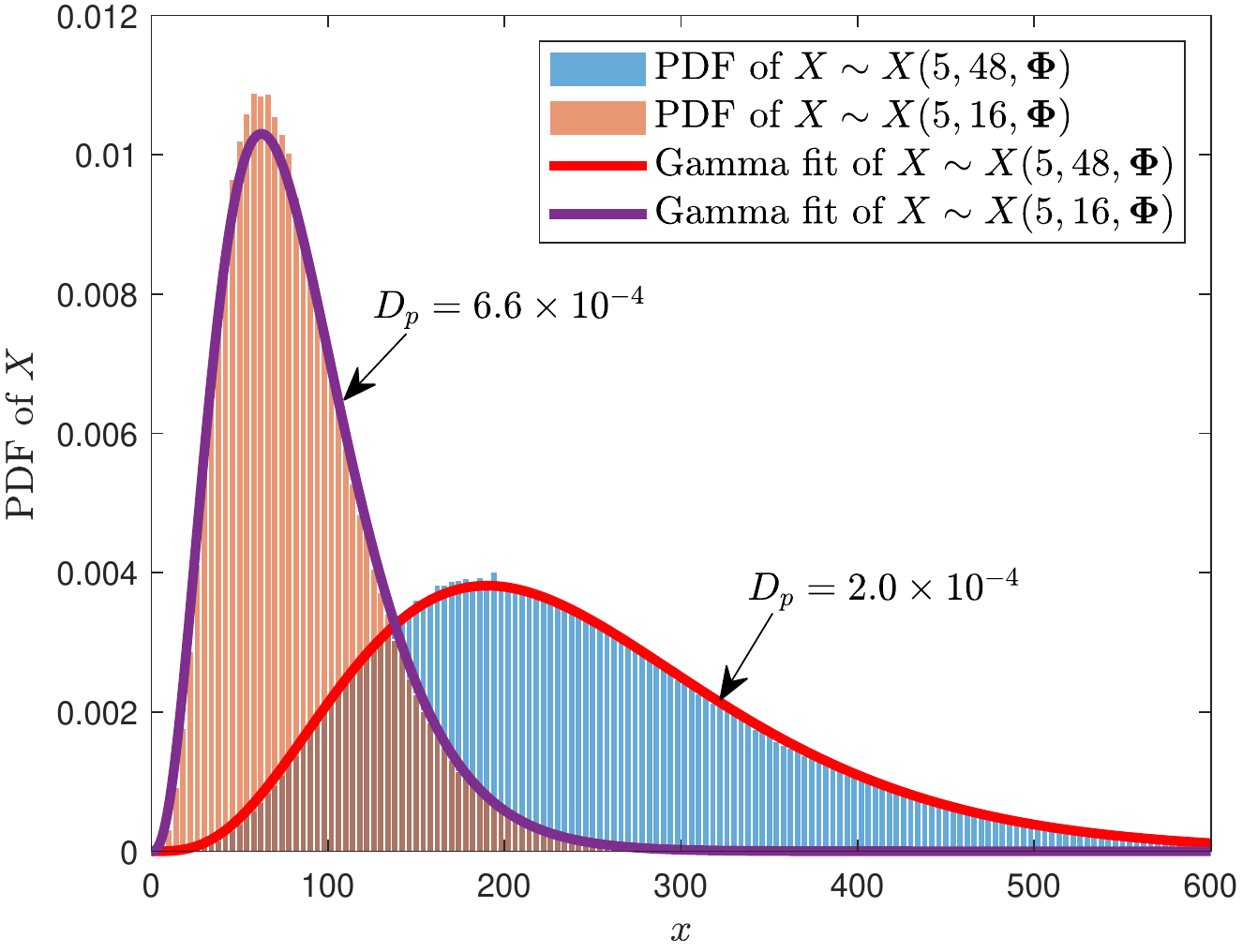}
\caption{KS test for Gamma distributions and real distributions of $X\sim X(N_e,N,\bm{\Phi})$ via $10^5$ Monte Carlo simulations, where $\bm{\Phi}$ is an arbitrary diagonal matrix as defined in Lemma 1.}
\label{KS}
\end{figure}

With Lemma 1, we deduce the expression of ergodic secrecy rate with optimal phase shift matrix $\bm{\Phi}_i^*$. The ergodic secrecy rate is defined as follows \cite{Liu2020},
\begin{flalign}\label{esr}
R_i& = [\mathbb{E}(C_{m,i}|\bm{\Phi}^*_i)-\mathbb{E}(C_{w,i}|\bm{\Phi}^*_i)]^+ \\ 
&\leq \mathbb{E}[(C_{m,i}-C_{w,i})^{+}|\bm{\Phi}^*_i], \label{realesr}
\end{flalign}
where $\{=\}$ in Eq. (\ref{realesr}) holds if and only if the instantaneous secrecy rate $\{C_{m,i}-C_{w,i}\}$ is nonnegative in all CSI realizations. Since Eve's CSIs $\mathbf{g}_{i}$ and $\mathbf{Z}_i$ are unknown, it is hard to determine whether an instantaneous secrecy rate is nonnegative or not, so we use the lower bound of real ergodic secrecy rate as the performance metric for the optimization process, as shown in Eq. (\ref{esr}).

\begin{theorem}[Expression of ergodic secrecy rate]
The expression of ergodic secrecy rate of $U_i$ with $\bm{\Phi}^*_i$ can be expressed as
\begin{flalign}\label{eesr}
&\bar{R}_i=[\bar{C}_{m,i}-\mathbb{E}(C_{w,i}|\bm{\Phi}^*_i)]^+,
\end{flalign}
where 
\begin{flalign}
\mathbb{E}(C_{w,i}|\bm{\Phi}^*_i)=\frac{1}{\ln(2)\Gamma(\mu)}G_{2,3}^{3,1}\bigg( \frac{\sigma_{e,i}^2}{\nu P_i \alpha_{e,i}^2 }\bigg| \begin{matrix}0, 1 \\
0,0, \mu
\end{matrix} \bigg),
\end{flalign}
and $\mathbb{E}(C_{m,i}|\bm{\Phi}^*_i)=\bar{C}_{m,i}$. $\mu$ and $\nu$ are defined in Eqs. (\ref{mu}) and (\ref{nu}), respectively. Note that $\bar{C}_{m,i}$ is the ergodic channel capacity between $U_i$ and BS, which can be measured by the statistical method as the CSIs of $U_i$, i.e., $\mathbf{l}_i$, $\mathbf{A}_i$, and $\mathbf{h}_i$ in each transmission burst can be obtained perfectly.
\end{theorem}

\begin{proof}
See Appendix B. 
\end{proof}

\section{Gas-Oriented Computation Offloading Scheme}\label{proposed2}
The proposed computation offloading scheme includes two processes, i.e., computational resource allocation and data uploading. The secrecy rates $R_i, \forall i$ in Eq. (\ref{ec}) are time-variant as the coherence time is much shorter than the computation offloading duration. In this case, we use ergodic secrecy rate $\bar{R}_i$ in Theorem 1 as the metric to allocate computational resource, i.e., use $\bar{R}_i, \forall i$ to calculate $Q_{i,k}$ in Eq. (\ref{ec}) instead of $R_i, \forall i$. In the data uploading process, the CSIs of legitimate users are time-variant and can be obtained via channel estimation in each transmission burst, so adaptive PLS coding should be designed to improve the efficiency of secure transmissions.

\subsection{Computational Resource Allocation}

\subsubsection{Bipartite graph generation} In order to keep the degrees of dissatisfaction of each sensor below a threshold, i.e., constraint (\ref{p1c1}), $N_I$ sensors are sort as $\{ U_{1'},U_{2'},...,U_{N_I'}\}$ where their Gas obey $V_{1'}\geq V_{2'} \geq ... \geq V_{N_I'}$. Similarly, the MEC servers are sort as $\{ M_{1''},M_{2''},...,M_{N_I''}\}$ where their computational power obey $f_{1''}/c_{1''}\geq f_{2''}/c_{2''} \geq ... \geq f_{N_I''}/c_{N_I''}$ by dropping $N_K-N_I$ weaker MEC servers. We define $\mathcal{S}=\{1',...,N_I'\}$ and $\mathcal{M}=\{1'',...,N_I''\}$ as the indices of the re-ordered sets of sensors and MEC servers, respectively. 

Since secrecy rate is time-varying during the computation offloading process, $Q_{i,k}$ defined in Eq. (\ref{ec}) can not be measured. Instead of the secrecy rate, Theorem 1 provides ergodic secrecy rate that is used to calculate $Q_{i,k}$ for each pair of a sensor and an MEC server as follows,
\begin{flalign} 
Q_{i,k} = \eta_k c_k D_i  f_k^2+\frac{D_i P_i}{\bar{R}_iB},  i\in \mathcal{S}, k\in \mathcal{M},
\end{flalign}
and then a bipartite graph is generated to record all sensor-MEC combinations, as shown in Fig. \ref{bip}. To meet constraint (\ref{p1c1}), the case of $r(V_i)-r(f_k/c_k)> \epsilon$ for any sensor-MEC pair should not be recorded in the bipartite graph. Hence, if $|i-k|>\epsilon$, $i\in \mathcal{S}, k\in \mathcal{M}$, $Q_{i,k}$ is set to be $\infty$.

\begin{figure}[htb]
\centering
\includegraphics[width=1\linewidth]{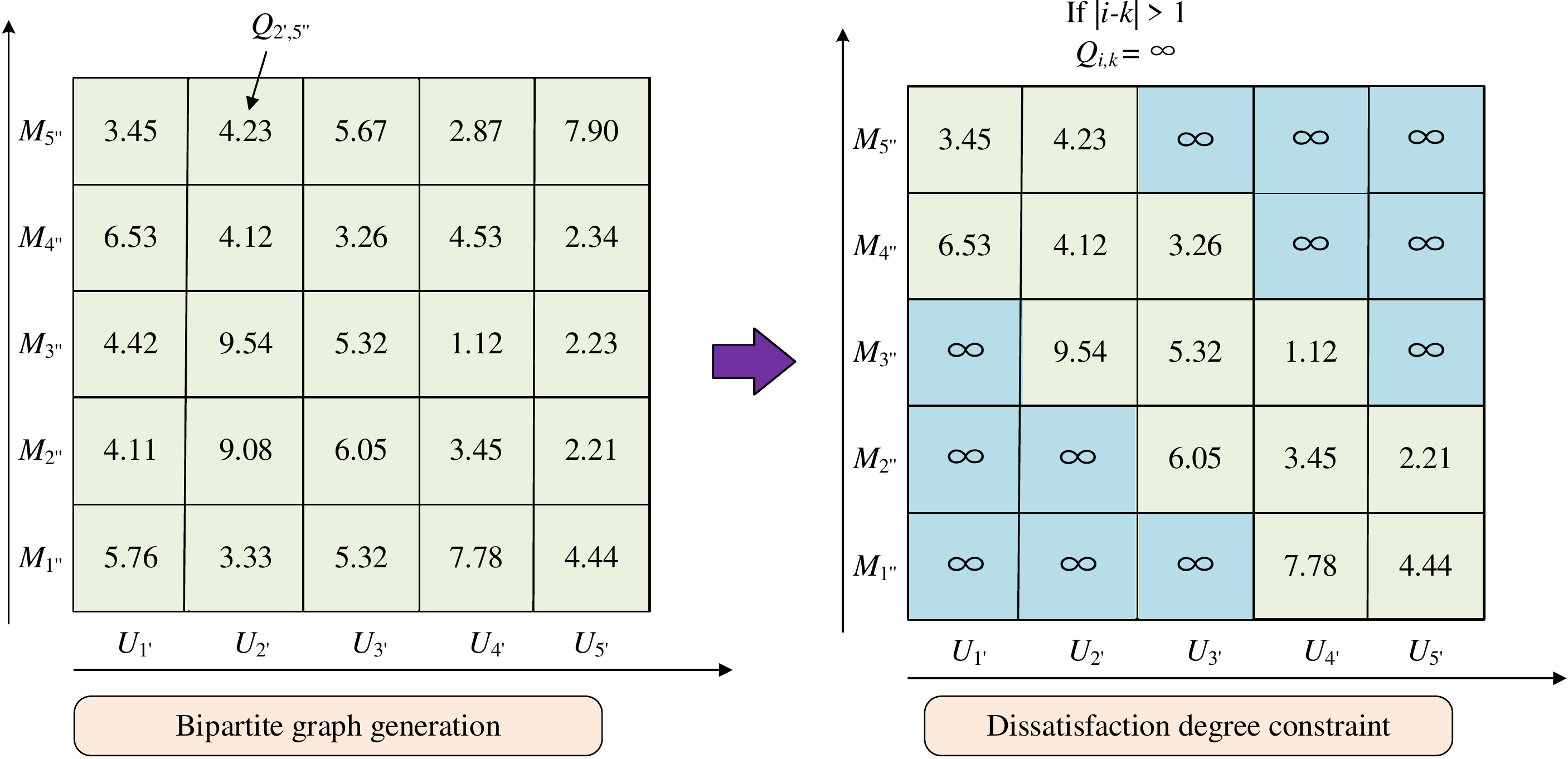}
\caption{Bipartite graph generation with the dissatisfaction degree constraint, where $\epsilon=1$.}
\label{bip}
\end{figure}

\subsubsection{Matching process} We set an $N_I\times N_I$ matching matrix $\mathbf{W}$, whose $(i,k)$-th element represents an edge, e.g., $W_{i,k}=1$ means that the vertexes $i$ and $k$ are linked. For another edge $W_{ii,kk}$, we have $i\neq ii$ and $k\neq kk$ to meet the constraint that any two edges do not share a vertex. With $\mathbf{W}$, the computational resource allocation problem can be equivalently transformed to a 2-dimensional matching problem as 
\begin{flalign}
\text{P4: }& \min_{\mathbf{W}}\sum_{i\in\mathcal{S}}^{N_I}\sum_{k\in \mathcal{M}}^{N_I} W_{i,k}Q_{i,k}, \\
& \text{s.t. } W_{i,k}= \{0 \text{ or }1\}, \forall i,k, \label{p3c1} \\
&\quad \; \sum_{i\in\mathcal{S}}^{N_I}W_{i,k}\leq 1, \; \sum_{k\in \mathcal{M}}^{N_I}W_{i,k}\leq 1. \label{p3c2}
\end{flalign}
It is obvious that $W_{i,k}$ is a binary variable taking 1 when $M_k$ is assigned to $U_i$, and 0 otherwise. Eq. (\ref{p3c2}) reveals that each sensor can be allocated with only one MEC server, and each MEC server is only assigned to one sensor. The optimization is to find the optimal $\mathbf{W}$ to minimize $\sum_{i\in\mathcal{S}}^{N_I}\sum_{k\in \mathcal{M}}^{N_I} W_{i,k}Q_{i,k}$. P4 is convex and the optimal $\mathbf{W}$ can be solved by the Kuhn-Munkres (KM) algorithm \cite{duan2012scaling}. 

\begin{figure*}[!htp]
\begin{subfigure}[t]{.5\textwidth}
\centering
\includegraphics[width=0.88\textwidth]{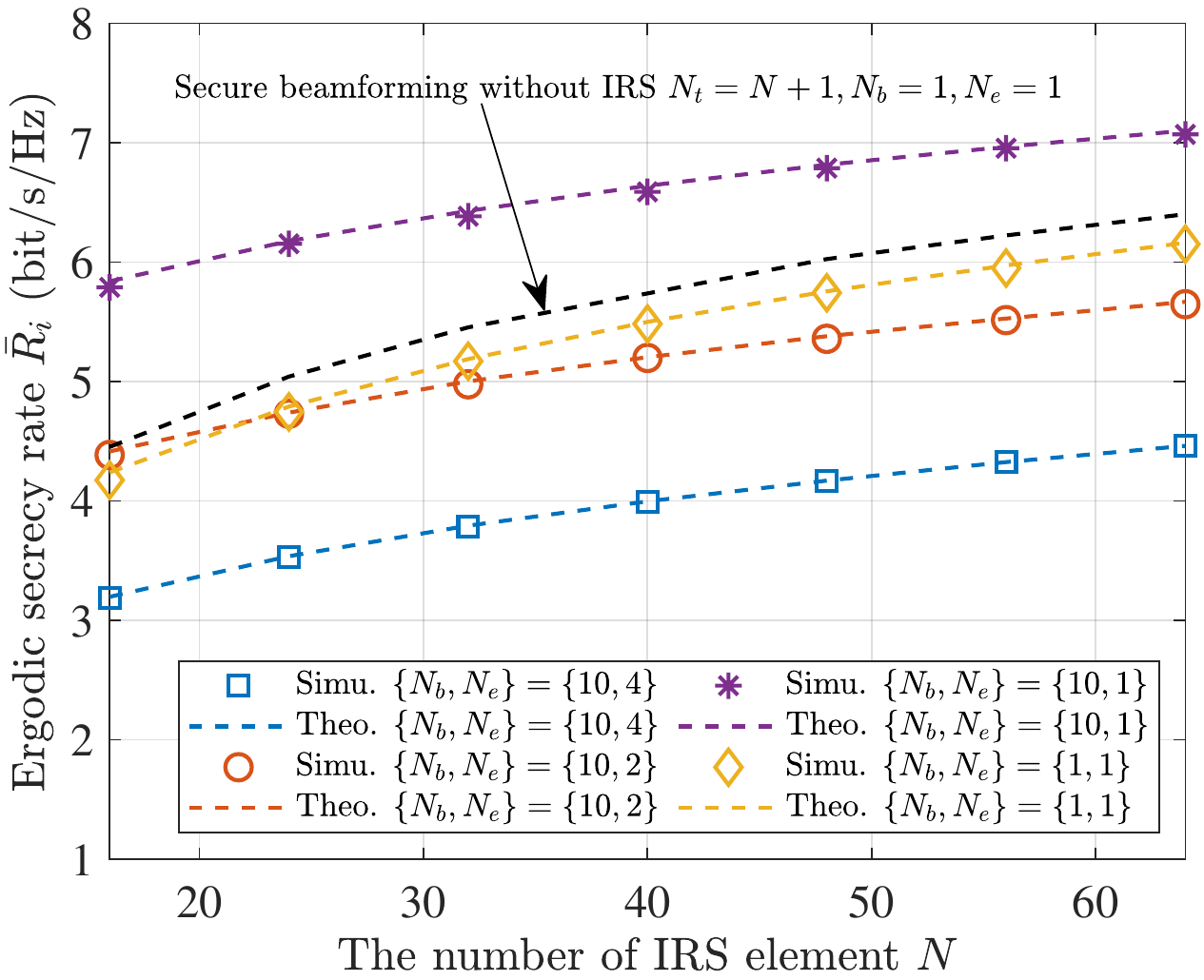}
\caption{Ergodic secrecy rate in terms of the number of IRS elements.}\label{sim1}
\end{subfigure}\hfill
\begin{subfigure}[t]{.5\textwidth}
\centering
\includegraphics[width=0.88\textwidth]{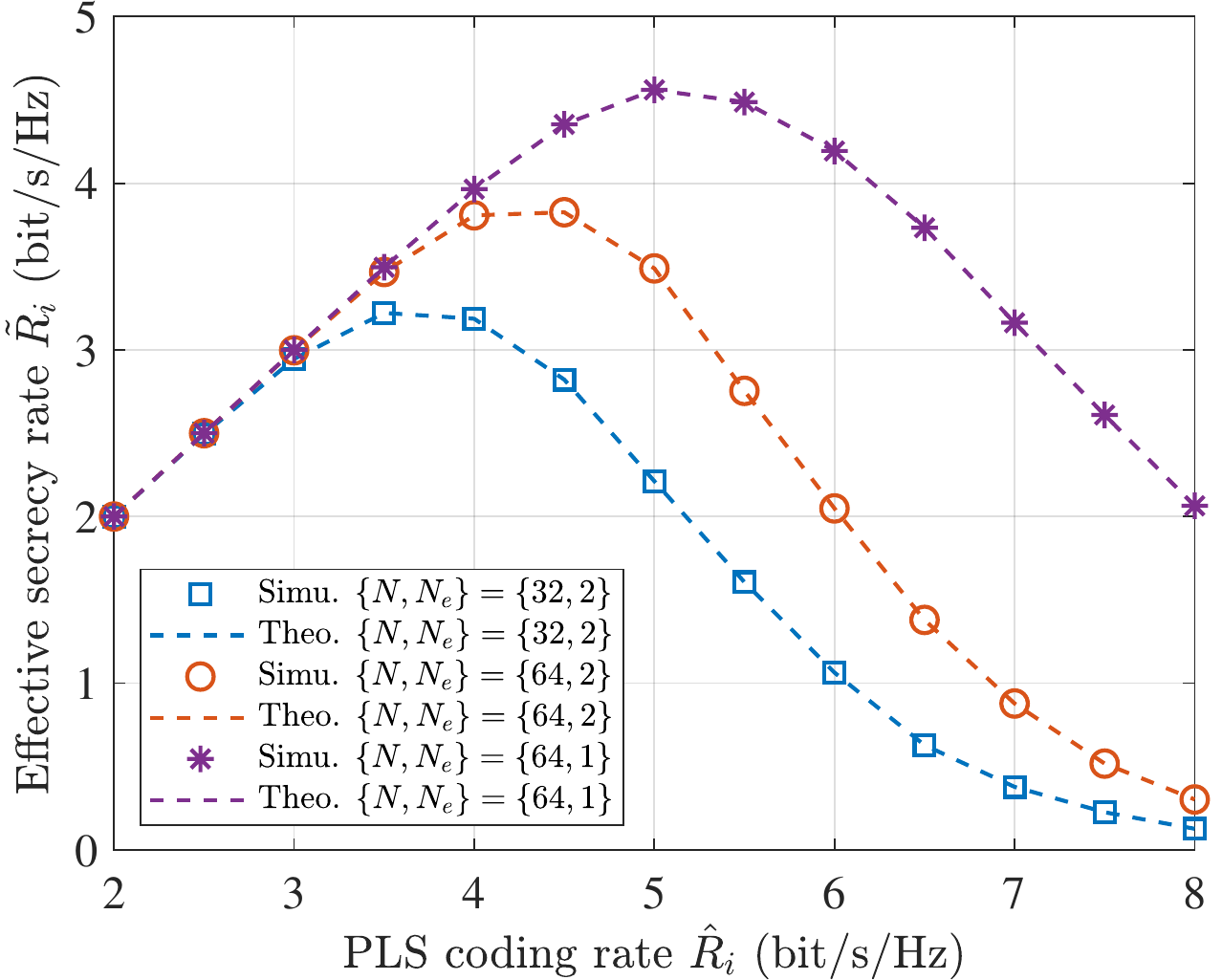}
\caption{Effective secrecy rate in terms of PLS coding rates with $N_b=10$.}\label{sim2}
\end{subfigure}
\caption{Theoretical results and Monte Carlo simulations of the ergodic secrecy rate and effective secrecy rate of $U_i$, where SNR = 63 dB.} \label{theorem}
\end{figure*}

\subsection{Adaptive PLS Coding}
The PLS coding was introduced first by Wyner \cite{Wyner75}, which is a code scheme to ensure that a coding rate (defined as the PLS coding rate) below the secrecy rate can achieve both reliability and information-theoretic security. In the data uploading process, $U_i$ uses the secrecy outage probability to adjust the PLS coding rate because Eve's CSI is unavailable for $U_i$, where the secrecy outage probability is regarded as the probability that the secrecy rate $R_i$ is smaller than the PLS coding rate of $U_i$, defined by $\hat{R}_i$. From \cite[Eq. (38)]{Liu2021}, the secrecy outage probability is given as 
\begin{flalign}\label{uso}
P(\hat{R}_i)&=P(R_i\leq \hat{R}_i \big| \text{decoding correctly}\text{ and }\bm{\Phi}^*_i) \notag \\
&=P(C_{w,i}\geq C_{m,i}- \hat{R}_i \big),
\end{flalign}
where $C_{w,i}$ and $C_{m,i}$ are defined in Eqs. (\ref{cm}) and (\ref{cw}) with $\bm{\Phi}^*_i$, respectively. Since the instantaneous CSI of $\mathbf{l}_i$, $\mathbf{A}_i$, and $\mathbf{h}_i$ are assumed to be estimated perfectly in each transmission burst, $U_i$ can adaptively adapt $C_{m,i}$ such that the codebook can be decoded correctly. The remaining work is to find an optimal $\hat{R}_i$ with the consideration of secrecy outage probability.

We use the effective secrecy rate as a secrecy metric to optimize $\hat{R}_i$, as investigated in \cite{Liu2021}. The effective secrecy rate is an average rate secretly received at BS over many transmission bursts with the pre-defined $\hat{R}_i$, which is expressed as
\begin{flalign}\label{effsr}
\tilde{R}_i=\{1-P(\hat{R}_i)\}\hat{R}_i.
\end{flalign}
In order to maximize $\tilde{R}_i$ by adjusting $\hat{R}_i$, the next work is to deduce the expression of effective secrecy rate, as shown in Theorem 2. 

\begin{theorem}[Effective secrecy rate] 
The expression of effective secrecy rate $\tilde{R}_i$ with the pre-defined $\hat{R}_i$ can be expressed as
\begin{flalign}
\tilde{R}_i(\hat{R}_i)&=\bigg[1-\frac{1}{\Gamma(N_e)}\Gamma\bigg(N_e,\frac{\phi}{1+|\bm{\Phi}_i^*\mathbf{h}_i|^2}\bigg)\bigg]\hat{R}_i,
\end{flalign}
where $\phi= \sigma_{e,i}^2(2^{C_{m,i}-\hat{R}_i}-1)/(\alpha_{e,i}^2P_i)$.
\end{theorem}

\begin{proof}
As $P(\hat{R}_i)=1/\Gamma(N_e)\Gamma(N_e,\phi/(1+|\bm{\Phi}_i^*\mathbf{h}_i|^2)$ \cite[Eq. (23)]{Liugit2}, it is easy to get the expression of effective secrecy rate by substituting $P(\hat{R}_i)$  into Eq. (\ref{effsr}).
\end{proof}

According to Theorem 2, the adaptive PLS coding scheme is designed to maximize the effective secrecy rate, i.e.,
\begin{flalign}
(\hat{R}_i^*)=\max_{\hat{R}_i,P(\hat{R}_i)\leq \varepsilon }\tilde{R}_i(\hat{R}_i),
\end{flalign}
where the optimal PLS coding rate $\hat{R}_i$, i.e., $\hat{R}_i^*$ can be found via the Golden-section search over the function $\tilde{R}_i(\hat{R}_i)$ as the function $\tilde{R}_i(\hat{R}_i)$ with respect to $\hat{R}_i$ is unimodal, as examined in Fig. \ref{sim2}. Note that the secrecy outage probability $P(\hat{R}_i)$ should be smaller than a threshold $\varepsilon$ to avoid a large secrecy outage probability in the transmission.

\subsection{Computational Complexity Analysis}

In the phase shift matrix optimization process, the computational complexity to get $\bm{\Phi}_i^*, \forall i$ is $N_I(4N^4+6N^3+N)$. The computational resource allocation requires two bubble sort algorithms to get $\{ U_{1'},U_{2'},...,U_{N_I'}\}$ and $\{ M_{1''},M_{2''},...,M_{N_I''}\}$, each of which needs $N_I^2$ iterations. The KM algorithm needs $N_I^4$ iterations to find the optimal matching $\mathbf{W}^*$ \cite{duan2012scaling}. The computational complexity of the PLS coding scheme based on the Golden-section search algorithm is $O(\log(1/\varepsilon))$, where $\varepsilon$ is the accuracy requirement.  In total, the computational complexity of the offloading process is $O[N_I(4N^4+6N^3+N)+N_I^4+2N_I^2+\log(1/\varepsilon)]$. 

\section{Simulations}\label{simulations}

The global simulation parameters are described as follows. The path loss parameter $\alpha_i$ is calculated by $\alpha_i=\frac{\tau}{2\pi f_c d_i}$, where $\tau$ is the speed of light, $f_c$ is work spectrum that is set to 2.4 GHz, $d_i, \forall i$ is the distance between the BS and $U_i$ is set to be uniform distribution over $[$30 m, 50 m$]$. $f_k, \forall k$, $D_i, \forall i$, and $V_i, \forall i$ are uniform distributions over $[$40 GHz, 80 GHz$]$, $[$610 KB, 1.8 MB$]$, and $[1.5\times 10^6$, $2\times 10^6]$. The AWGN floor parameters $\sigma_i^2$ and $\sigma_{e,i}^2$ are -53 dBm \cite{Dinc2016}. The transmission power $P_i, \forall i$ is 10 dBm. The computation energy efficiency coefficient $\eta_k, \forall k$ is set to be $10^{-27}$. Note that all simulation results are average values from $10^5$ independent runs. The available radio resource block (RB) has 150 KHz bandwidth. If two RBs are used for transmission, we have $B=300$ KHz.

Here, we examine Theorems 1 and 2 in Figs. \ref{sim1} and~\ref{sim2}, respectively. These figures show the good agreements between theoretical results (Theo.) and Monte Carlo (MC.) simulation results from $10^5$ independent runs. From Fig. \ref{sim1}, we can find that the ergodic secrecy rate increases with the increasing number of IRS elements. Also, a comparison is taken between IRS-assisted PLS and the secure beamforming scheme that uses $N_t=N+1$ transmission antennas in the case of $N_b=1$ and $N_e=1$ \cite{Liu2020}. From the comparison simulations, we find that the IRS-assisted PLS outperforms the secure beamforming scheme because the channel gain by IRS is larger than that of the beamforming scheme \cite{VanChien2021}. Fig. \ref{sim1} also shows that the number of Eve's antennas has a strong negative effect on the ergodic secrecy rate. Fig. \ref{sim2} shows the effective secrecy rate in terms of PLS coding rate. It is obvious that there is one peak in each curve, which is consistent with the results shown in \cite{Liu2021}, meaning that the unimodal function-aimed search algorithms, such as golden-section search, can be used to find the optimal PLS coding rate to maximize the effective secrecy rate.

\begin{figure}[t]
\centering
\includegraphics[width=0.90\linewidth]{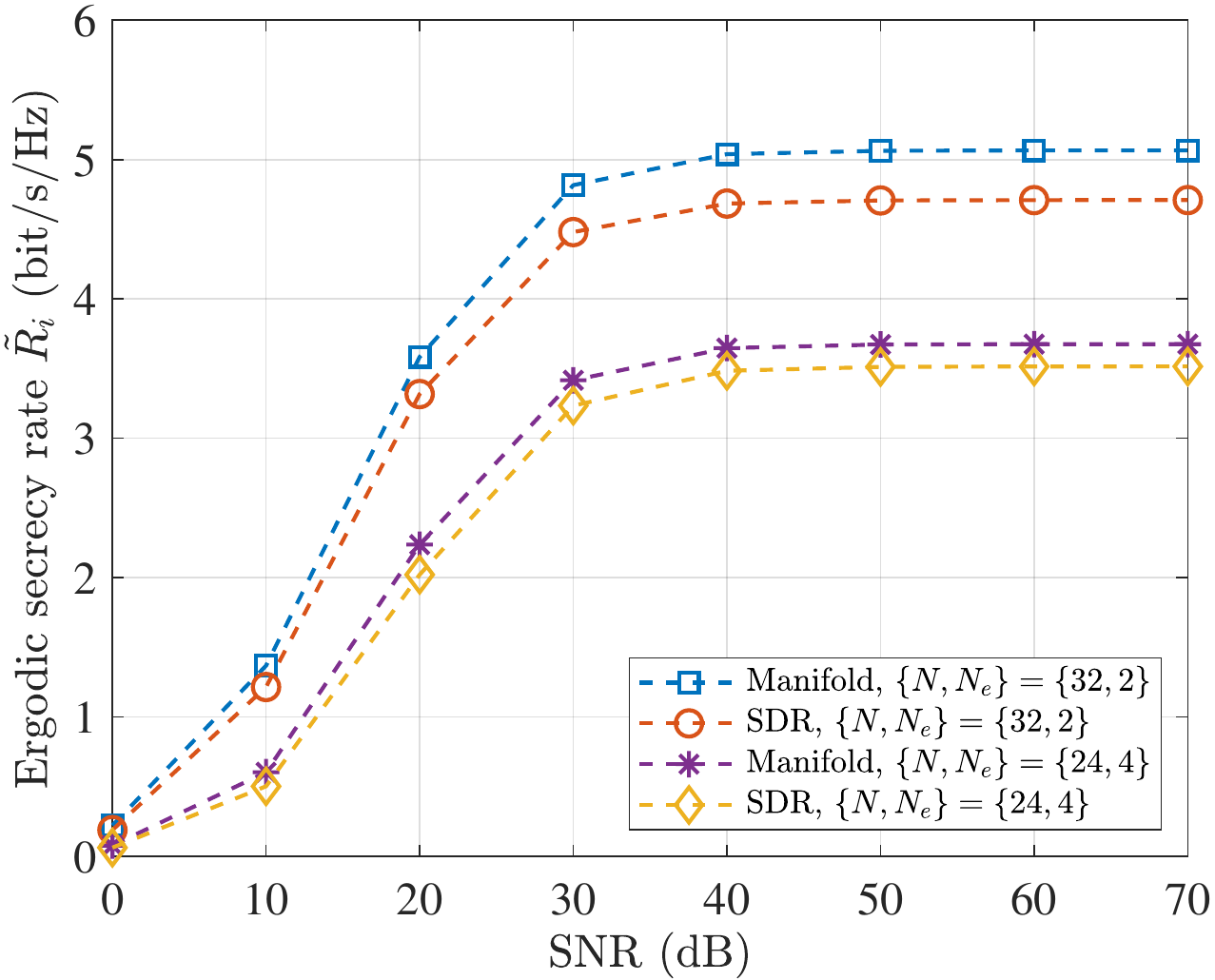}
\caption{Comparison simulations between the SDR-based scheme and the manifold optimization scheme.}
\label{sim4}
\end{figure}

It is hard to obtain the global optimal result of $\bm{\Phi}_i$ in P2 via globally searching algorithms as $\bm{\Phi}_i$ is in an $N$-continuous dimensional space. Hence, the common way to solve P2 is based on SDR or manifold optimization, such as \cite{Wu2019} and~\cite{Yuiccc2019}. The comparison simulations between the SDR-based scheme and the manifold optimization scheme are shown in Fig.~\ref{sim4}. From this figure, we can find the manifold optimization has better performance than that in the SDR-based scheme, because the rank one constraint results in that the SDR-based scheme just provides a near-optimal result for $\bm{\Phi}_i$ while the manifold optimization scheme can give the local optimal result~\cite{Wu2019}.

\begin{figure}[t]
\centering
\includegraphics[width=0.90\linewidth]{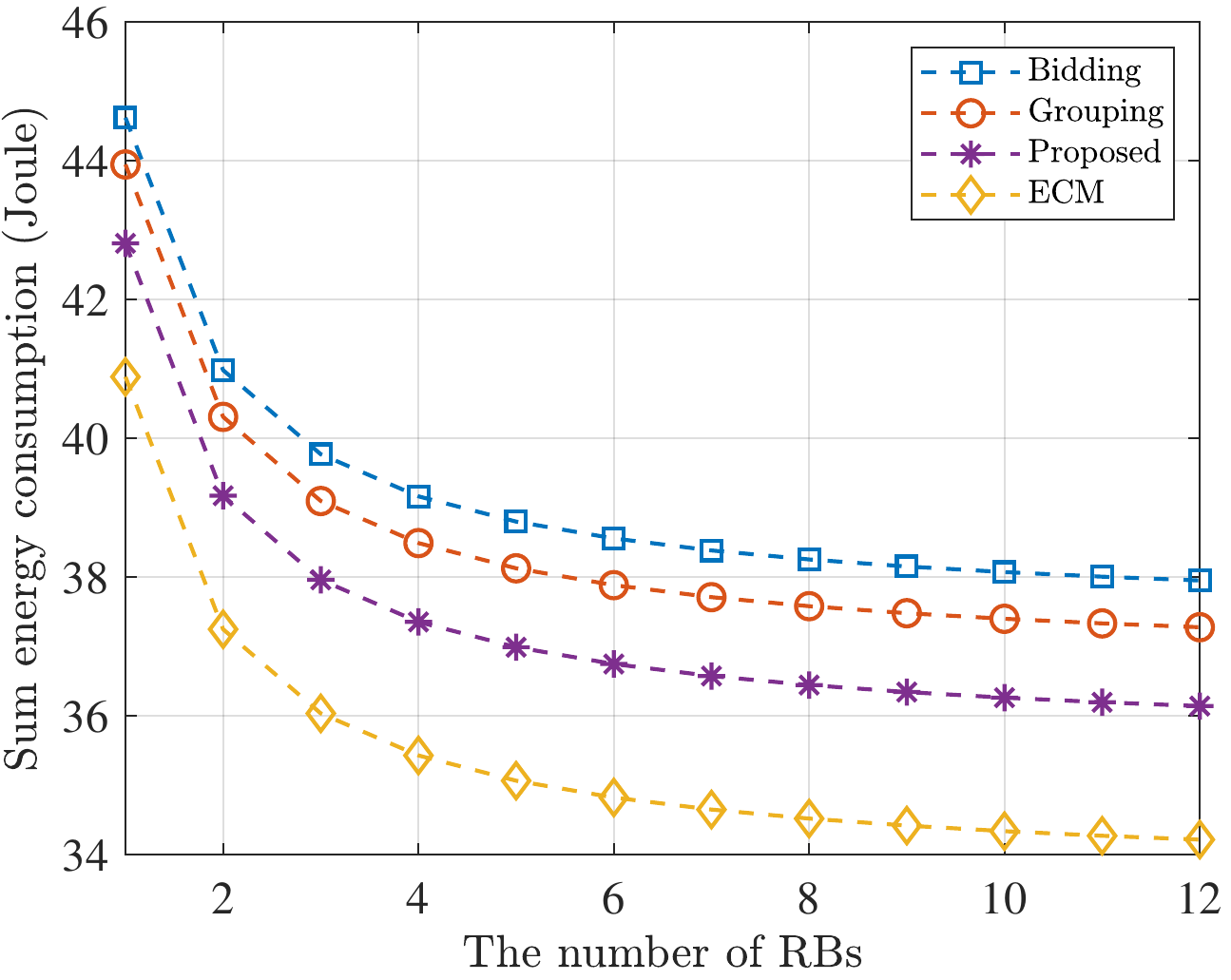}
\caption{Sum energy consumption of 40 sensors in terms of the number of RBs, where  $N=40$, $c_k=10,\forall k$, $\varepsilon = 0.1$, and $\epsilon=8$.}
\label{comp1}
\end{figure}

\begin{figure}[htb]
\centering
\includegraphics[width=0.90\linewidth]{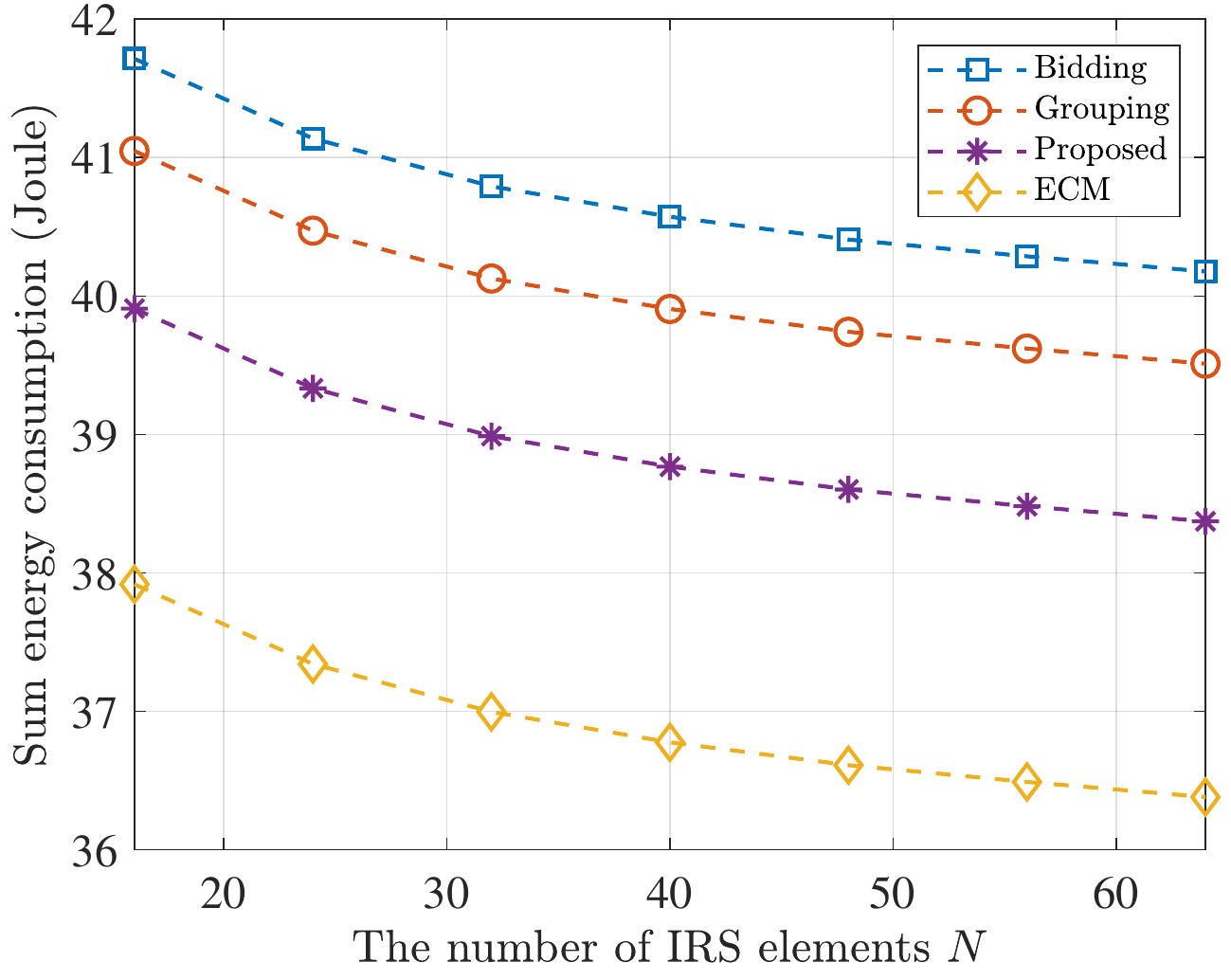}
\caption{Sum energy consumption of 40 sensors in terms of the number of IRS elements, where the number of RBs is 2, $c_k=10,\forall k$, $\varepsilon = 0.1$, and $\epsilon=8$.}
\label{comp2}
\end{figure}

\begin{figure}[htb]
\centering
\includegraphics[width=0.90\linewidth]{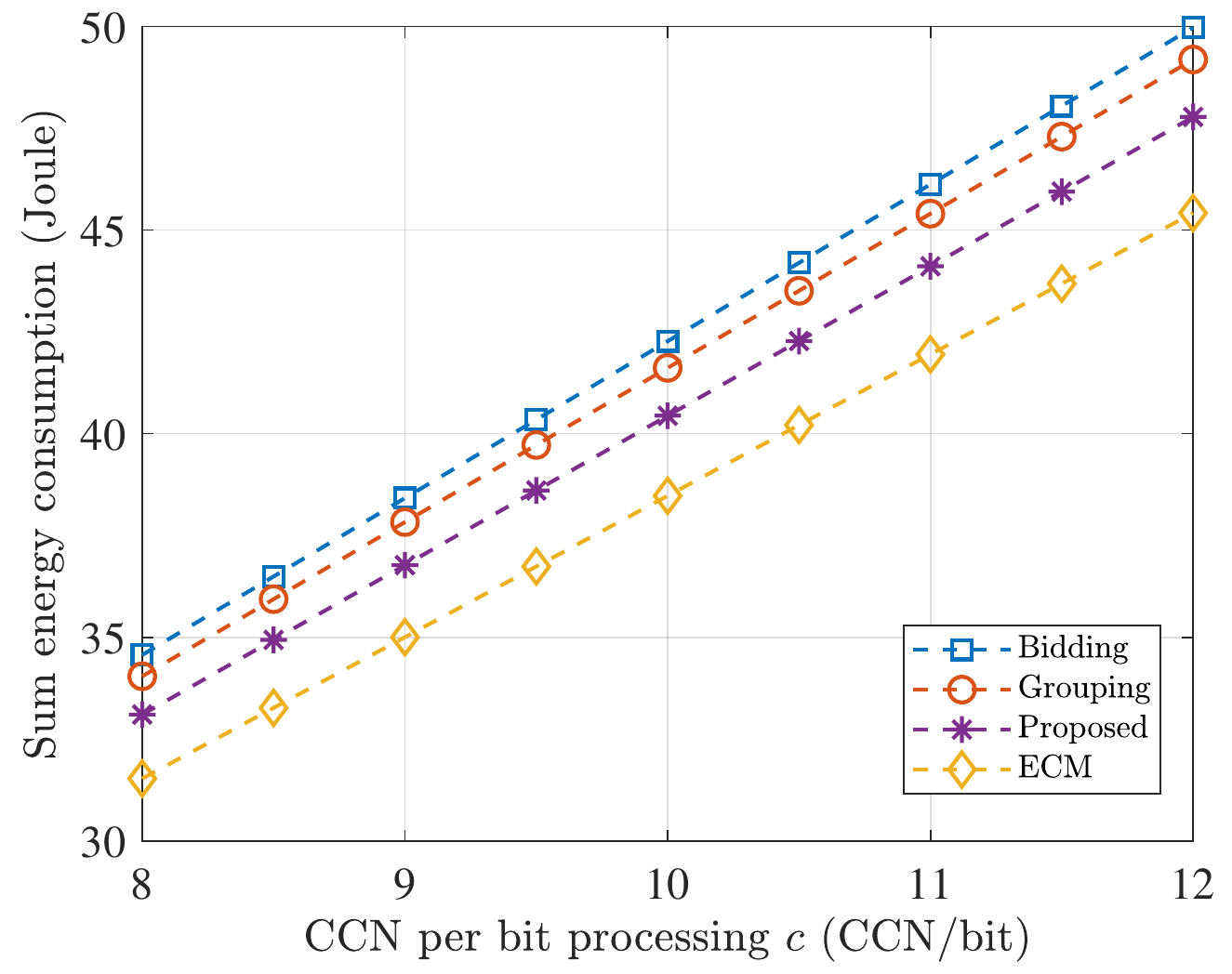}
\caption{Sum energy consumption of 40 sensors in terms of CCN per bit processing ($c=c_k,\forall k$), where the number of RBs is 2, $N=40$, $\varepsilon = 0.1$, and $\epsilon=8$.}
\label{comp3}
\end{figure}

\begin{figure}[htb]
\centering
\includegraphics[width=0.90\linewidth]{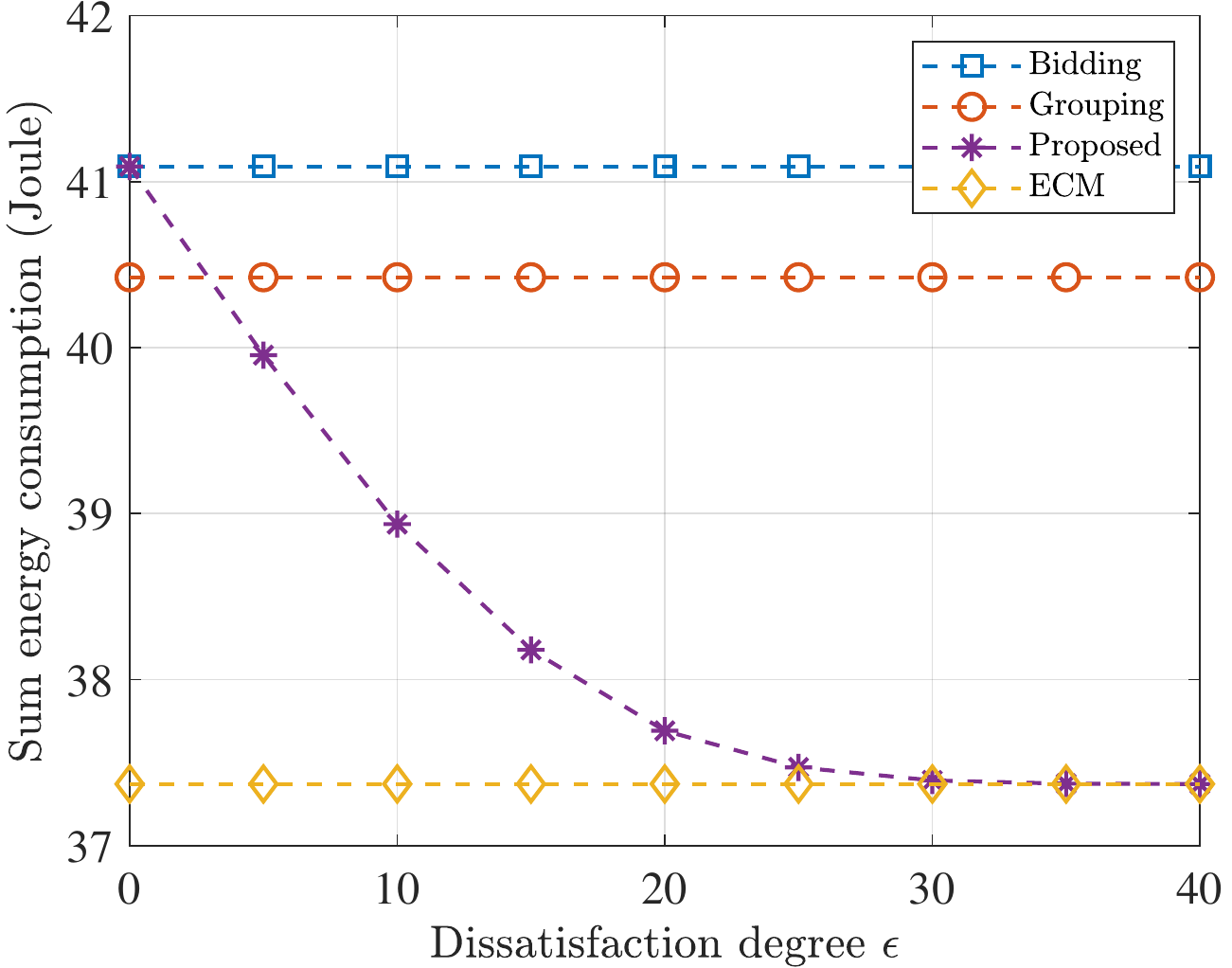}
\caption{Sum energy consumption of 40 sensors in terms of the dissatisfaction degree, where the number of RBs is 2, $N=40$, $c_k=10,\forall k$, and $\varepsilon = 0.1$.}
\label{comp4}
\end{figure}

The simulations of computational resource allocation are executed in terms of the number of RBs, the number of IRS elements $N$, CNN per bit processing $c_k$, and dissatisfaction degree $\epsilon$. The comparison simulations have been done among four different allocation schemes, i.e., bidding, grouping, energy consumption minimization (ECM), and the proposed scheme, which are described detailedly as follows.

\begin{enumerate}
\item Bidding scheme: It is a traditional strategy that the sensor who pays the highest Gas will get the best computational resource among the MEC servers. This scheme  requires two bubble sort algorithms and needs $O[N_I(4N^4+6N^3+N)+2N_I^2+\ln(1/\varepsilon)]$ computational complexity.
\item Group scheme. It is proposed in \cite{Liugit1} where sensors are divided into different groups based on paid Gas, and the group with higher Gas is prioritized with better computational resources. The computational complexity of this scheme is $O[N_I(4N^4+6N^3+N)+2N_I^2+N_G(N_m+1)^4+\ln(1/\varepsilon)]$, where $N_G$ and $N_m$ are the numbers of groups and sensors in each group, respectively.
\item ECM scheme: In the ECM scheme, the bipartite graph is generated to record the $Q_{i,k}$ for all sensor-MEC combinations, then, the KM algorithm is used to find the optimal matching to minimize the  energy consumption, as used in  \cite{Mu2020,Han2021iot,Dai2018}. The computational complexity of this scheme is $O[N_I(4N^4+6N^3+N)+2N_I^2+N_I^4+\log(1/\varepsilon)]$.
\item Proposed scheme: This paper uses joint phase shift matrix optimization and Gas-oriented computation offloading to reduce the energy consumption while keeping the dissatisfaction degrees of sensors are below a threshold. The computational complexity is similar to the ECM scheme, as discussed in Section V. D.
\end{enumerate}

From Fig. \ref{comp1}, we can check the number of RBs on the sum energy consumption of 40 sensors, where $\epsilon=8$, $N=40$, and $c=10$. It is demonstrated that the sum energy consumption decreases with the increasing number of RBs, because data transmission time is reduced with more bandwidth. Then, we find the increasing trend becomes slow when the bandwidth is enough as the global performance should consider other factors, such as computing power. Also, the proposed scheme outperforms the bidding and group schemes, but has more energy consumption than the ECM scheme. 

It is demonstrated that the increasing number of IRS elements will reduce the sum energy consumption, as shown in Fig. \ref{comp2}. However, we find the decline is slow because the computation offloading system follows the ``cask principle'', meaning that unilaterally improving the capability of communication modules cannot give a large gain for the performance of computation offloading systems. Fig. \ref{comp3} shows the sum energy consumption in terms of CCN per bit processing. There exists a rough linear rise for sum energy consumption with the CCN per bit processing. As the CCN per bit processing is the main metric for the algorithm efficiency of CPU chips, it implies that the best to reduce energy consumption is the improvement of the efficiency of CPU designs.

Fig. \ref{comp4} is used to check the dissatisfaction degree of sensors on the sum energy consumption. It is obvious that if the sensor has the highest dissatisfaction, the proposed scheme has the same performance as that in the ECM scheme, meaning that the optimization objective of the proposed scheme is to minimize the sum energy consumption and the satisfaction of sensors is no longer considered. Also, if $\epsilon=0$, the sensor gets the computational resource whose ranking should be equivalent to the ranking of the paid Gas among all sensors. Hence, the performance of $\epsilon=0$ is the same as the bidding scheme. Dissatisfaction degrees are not considered in bidding, grouping, and ECM schemes, so the simulation lines keep constant in this figure.

\section{Conclusions}\label{conclusions}
In this article, we establish a blockchain-empowered computation offloading system in an IoT network, where the data is protected by IRS-assisted PLS methods. Especially, we deduce the expression of the ergodic secrecy rate of the scheme for the computational resource allocation. Also, we design a Gas-oriented energy consumption minimization algorithm where higher Gas providers are prioritized with better computational resources, while improving the effective secrecy rate in the data uploading process. Simulations show that the proposed scheme reduces the energy consumption while guaranteeing that the node paying higher Gas has more opportunities to get a stronger computational resource. In future works, we will integrate the promising space division multiple access technologies into this system to further reduce the latency of the computation offloading process.

\newcounter{mytempeqncnt}
\begin{figure*}[t]
\begin{flalign}\label{meiG}
\int_{0}^{\infty}z^{\rho-1}(z+\nu)^{-\eta}G_{pq}^{mn}\bigg( \mu z \bigg| \begin{matrix}a_1, \dots, a_p \\
b_1, \dots, b_q 
\end{matrix} \bigg)dz=\frac{\nu^{\rho-\eta}}{\Gamma(\eta)}G_{p+1,q+1}^{m+1,n+1}\bigg( \mu \nu \bigg| \begin{matrix}1-\rho,a_1, \dots, a_p \\
\eta-\rho,b_1, \dots, b_q 
\end{matrix} \bigg)
\end{flalign}
\hrulefill
\vspace*{4pt}
\end{figure*}

\section*{Appendix}

\subsection{Proof of Lemma 1}\label{pol3}
Lemma \ref{nchis} is used to prove Lemma 1.
\begin{lemma}\label{nchis}[Proved in \cite{Patnaik1949}]
For $n$ independent Gaussian random variables $X_i\sim \mathcal{N}(\mu_i,\sigma^2)$, $Y=\sum_{i=1}^{n}X_i^2/\sigma^2$ obeys the noncentral chi-square distribution with degree of freedom $n$, and we have the mean and variance of $Y$ as follows,
\begin{flalign}\label{Le1}
\mathbb{E}(Y)=n+\lambda,\quad \mathbb{V}(Y)=2n+4\lambda,
\end{flalign}
where $\lambda=\sum_{i=1}^n\mu_i^2/\sigma^2$.
\end{lemma}

At first, we calculate the mean of $X\sim X(N_e,N,\bm{\Phi})$ as follows,
\begin{flalign}
\mathbb{E}(x)&=\mathbb{E}(|\mathbf{g}|^2)+\mathbb{E}(|\mathbf{Z}\bm{\Phi}\mathbf{h}|^2).
\end{flalign}
It is obvious that $\mathbb{E}(|\mathbf{g}|^2)=N_e$. 

Since $\mathbf{\Phi}$ is independent with $\mathbf{g}$, $\mathbf{Z}$, and $\mathbf{h}$, we introduce two auxiliary random variables $\mathbf{u}_1=\bm{\Phi}\mathbf{h}$ and $\mathbf{u}_2=\mathbf{Z}\bm{\Phi}\mathbf{h}/|\bm{\Phi}\mathbf{h}|$ with $\mathbf{u}_1\sim\mathcal{CN}_{N,1}(\mathbf{0},\mathbf{I}_{N})$ and $\mathbf{u}_2\sim\mathcal{CN}_{N_e,1}(\mathbf{0},\mathbf{I}_{N_e})$, respectively. It is obvious that $\mathbf{u}_1$ and $\mathbf{u}_2$ are independent, thus, $\mathbb{E}(|\mathbf{Z}\bm{\Phi}\mathbf{h}|^2)$ can be transformed to
\begin{flalign}
\mathbb{E}(|\mathbf{Z}\bm{\Phi}\mathbf{h}|^2)&=\mathbb{E}\bigg(\bigg|\frac{\mathbf{Z}\bm{\Phi}\mathbf{h}}{|\bm{\Phi}\mathbf{h}|}\bigg|^2\bigg)\mathbb{E}(|\bm{\Phi}\mathbf{h}|^2) \overset{a}{=}N_eN,
\end{flalign}
where $\overset{a}{=}$ is achieved because $\mathbb{E}(|\frac{\mathbf{Z}\bm{\Phi}\mathbf{h}}{|\bm{\Phi}\mathbf{h}|}|^2)=N_e$ and $\mathbb{E}(|\bm{\Phi}\mathbf{h}|^2)=N$. Hence, the mean of $x$ can be re-written as 
\begin{flalign}
\mathbb{E}(x)&=N_e(1+N).
\end{flalign}
Then, we will deduce the variance of $x$, i.e., $\mathbb{V}(x)$, which is given as 
\begin{flalign}\label{var}
\mathbb{V}(x)=\mathbb{E}(|x|^2)-|\mathbb{E}(x)|^2,
\end{flalign}
where $\mathbb{E}(|x|^2)$ can be expressed as
\begin{flalign}
&\mathbb{E}(x^2)  \notag \\
&=\mathbb{E}(\big||\mathbf{g}+\mathbf{Z}\bm{\Phi}\mathbf{h}|^2\big|^2) \notag \\ 
&=\mathbb{E}(\big||\mathbf{g}|^2+|\mathbf{Z}\bm{\Phi}\mathbf{h}|^2+  \mathbf{h}^{\rm{H}}\bm{\Phi}^{\rm{H}}\mathbf{Z}^{\rm{H}}\mathbf{g}+\mathbf{g}^{\rm{H}} \mathbf{Z}\bm{\Phi}\mathbf{h}\big|^2) \notag \\ 
&=\mathbb{E}(\big||\mathbf{g}|^2\big|^2)+\mathbb{E}(\big||\mathbf{Z}\bm{\Phi}\mathbf{h}|^2\big|^2)+\mathbb{E}(|\mathbf{h}^{\rm{H}}\bm{\Phi}^{\rm{H}}\mathbf{Z}^{\rm{H}}\mathbf{g} |^2) \notag \\
&+\mathbb{E}(|\mathbf{g}^{\rm{H}} \mathbf{Z}\bm{\Phi}\mathbf{h}|^2)+2\mathbb{E}(| \mathbf{g}|^2|\mathbf{Z}\bm{\Phi}\mathbf{h}|^2).
\end{flalign}
According to the property of noncentral chi-square distribution defined in Lemma 4, we have
\begin{flalign}
\mathbb{E}(\big||\mathbf{g}|^2\big|^2)&=\mathbb{V}(|\mathbf{g}|^2)+|\mathbb{E}(|\mathbf{g}|^2)|^2\notag \\
&=N_e^2+N_e. 
\end{flalign}
Also, we can get $\mathbb{E}(| \mathbf{g}|^2|\mathbf{Z}\bm{\Phi}\mathbf{h}|^2)=N_e^2N$, $\mathbb{E}(| \mathbf{h}^{\rm{H}}\bm{\Phi}^{\rm{H}}\mathbf{Z}^{\rm{H}}\mathbf{g}|^2)=\mathbb{E}(| \mathbf{g}^{\rm{H}} \mathbf{Z}\bm{\Phi}\mathbf{h}|^2)=N_eN$, and 
\begin{flalign}
\mathbb{E}(|\mathbf{Z}\bm{\Phi}\mathbf{h}|^4)&=\mathbb{E}\bigg(\bigg|\frac{\mathbf{Z}\bm{\Phi}\mathbf{h}}{|\bm{\Phi}\mathbf{h}|}\bigg|^4\bigg)\mathbb{E}(|\bm{\Phi}\mathbf{h}|^4), \notag \\
&=(N_e+N_e^2)(1+N)N.
\end{flalign}
Then, according to Eq. (\ref{var}), we have the variance of $x$ as follows,
\begin{flalign}
\mathbb{V}(x)=N_e(1+N)^2+(N_e+N_e^2)N.
\end{flalign}
According to the definition of the Gamma distribution \cite{weisstein2002crc}, the shape and scale of the Gamma distribution can be expressed as
\begin{flalign}
&\mu=\frac{[\mathbb{E}(x)]^2}{\mathbb{V}(x)}, \quad \nu=\frac{\mathbb{V}(x)}{\mathbb{E}(x)},
\end{flalign}
then, we get the PDF and CDF of $X\sim X(N_e,N,\bm{\Phi})$ as defined in Eqs. (\ref{pdf}) and (\ref{cdf}).

The proof is completed. \hfill $\IEEEQEDclosed$

\subsection{Proof of Theorem 2}
Here, we deduce the ergodic capacity between $U_i$ and Eve $\mathbb{E}(C_{w,i}|\bm{\Phi}^*_i)$ as follows. Since $\mathbf{\Phi}_i^*$ is related with CSIs of legitimate users and is independent with $\mathbf{g}_i$, $\mathbf{Z}_i$, and $\mathbf{h}_i$, $\mathbb{E}(C_{w,i}|\bm{\Phi}^*_i)$ can be expressed as
\begin{flalign}\label{t2a1}
\mathbb{E}(C_{w,i}|\bm{\Phi}^*_i)&=\int_0^{\infty}\log_2\bigg(1+\frac{\alpha_{e,i}^2P_i}{\sigma_{e,i}^2}x\bigg)f_{X}(x)dx \notag \\
&=\frac{1}{\ln(2)}\int_0^{\infty}\ln\bigg(1+\frac{\alpha_{e,i}^2P_i}{\sigma_{e,i}^2}x\bigg)f_{X}(x)dx \notag \\
&\overset{a}{=}\frac{1}{\ln(2)}\int_0^{\infty}\frac{1}{1+z}\bigg[1-F_{X}\bigg(\frac{\sigma_{e,i}^2z}{P_i\alpha_{e,i}^2}\bigg)\bigg]dz \notag \\
&=\frac{1}{\ln(2)\Gamma(\mu)}\int_0^{\infty}\frac{\Gamma(\mu,\frac{\sigma_{e,i}^2z}{\nu P_i\alpha_{e,i}^2})}{1+z}dz,
\end{flalign}
where $\overset{a}{=}$ is due to a change of variable $z=\frac{\alpha_{e,i}^2P_i}{\sigma_{e,i}^2}x$. $F_X(\cdot)$ is CDF of the variable $X\sim X(N_e,N,\bm{\Phi})$, as defined in Eq. (\ref{cdf}). $\mu$ and $\nu$ are defined in Eqs. (\ref{mu}) and (\ref{nu}).
According to~\cite[Eq. (5)]{Kumar2015}, we have
\begin{flalign}
\Gamma(a,x)=G_{1,2}^{2,0}\bigg( x \bigg| \begin{matrix} 1 \\
0, a
\end{matrix} \bigg).
\end{flalign}
Jointly considering Eq. (\ref{meiG}) \cite[7.812 EI II 418(4)]{Jeffrey2007}, we get
\begin{flalign}\label{t2a2}
&\int_0^{\infty}\frac{\Gamma(\mu,\frac{\sigma_{e,i}^2z}{\nu \alpha_{e,i}^2 P_i})}{1+z}dz \notag \\
&=\int_0^{\infty}(1+z)^{-1}G_{1,2}^{2,0}\bigg( \frac{\sigma_{e,i}^2z}{\nu P_i \alpha_{e,i}^2} \bigg| \begin{matrix} 1 \notag \\
0, \mu
\end{matrix} \bigg)dz \notag \\
&=G_{2,3}^{3,1}\bigg( \frac{\sigma_{e,i}^2}{\nu P_i \alpha_{e,i}^2}\bigg| \begin{matrix}0, 1 \\
0,0, \mu
\end{matrix} \bigg).
\end{flalign}
Substituting Eq. (\ref{t2a2}) into Eq. (\ref{t2a1}), we get the expression of the ergodic capacity of between $U_i$ and BS as follows,
\begin{flalign}\label{t2a3}
&\mathbb{E}(C_{w,i}|\bm{\Phi}^*_i)=\frac{1}{\ln(2)\Gamma(\mu)}G_{2,3}^{3,1}\bigg( \frac{\sigma_{e,i}^2}{\nu P_i \alpha_{e,i}^2 }\bigg| \begin{matrix}0, 1 \\
0,0, \mu
\end{matrix} \bigg).
\end{flalign}

Substituting Eq. (\ref{t2a3}) into Eq.(\ref{eesr}), we obtain the expression of the ergodic secrecy rate.

The proof is completed. \hfill $\IEEEQEDclosed$

\bibliographystyle{IEEEtran}
\bibliography{references}

\begin{thebibliography}{10}
\providecommand{\url}[1]{#1}
\csname url@samestyle\endcsname
\providecommand{\newblock}{\relax}
\providecommand{\bibinfo}[2]{#2}
\providecommand{\BIBentrySTDinterwordspacing}{\spaceskip=0pt\relax}
\providecommand{\BIBentryALTinterwordstretchfactor}{4}
\providecommand{\BIBentryALTinterwordspacing}{\spaceskip=\fontdimen2\font plus
\BIBentryALTinterwordstretchfactor\fontdimen3\font minus
  \fontdimen4\font\relax}
\providecommand{\BIBforeignlanguage}[2]{{%
\expandafter\ifx\csname l@#1\endcsname\relax
\typeout{** WARNING: IEEEtran.bst: No hyphenation pattern has been}%
\typeout{** loaded for the language `#1'. Using the pattern for}%
\typeout{** the default language instead.}%
\else
\language=\csname l@#1\endcsname
\fi
#2}}
\providecommand{\BIBdecl}{\relax}
\BIBdecl

\bibitem{Gozalvez2016}
J.~Gozalvez, ``New {3GPP} standard for {IoT} [mobile radio],'' \emph{IEEE Veh.
  Technol. Mag.}, vol.~11, no.~1, pp. 14--20, Mar. 2016.

\bibitem{Wang2021a}
Y.~Wang, Z.~Su, N.~Zhang, J.~Chen, X.~Sun, Z.~Ye, and Z.~Zhou, ``{SPDS}: {A}
  secure and auditable private data sharing scheme for smart grid based on
  blockchain,'' \emph{IEEE Trans. Ind. Inf.}, vol.~17, no.~11, pp. 7688--7699,
  Nov. 2021.

\bibitem{Liu2021}
Y.~Liu, W.~Wang, H.-H. Chen, F.~Lyu, L.~Wang, W.~Meng, and X.~Shen, ``Physical
  layer security assisted computation offloading in intelligently connected
  vehicle networks,'' \emph{IEEE Trans. Wireless Commun.}, vol.~20, no.~6, pp.
  3555--3570, Jun. 2021.

\bibitem{Su2020}
Z.~Su, Y.~Wang, Q.~Xu, and N.~Zhang, ``{LVBS}: {L}ightweight vehicular
  blockchain for secure data sharing in disaster rescue,'' \emph{IEEE Trans.
  Dependable Secure Comput.}, vol.~19, no.~1, pp. 19--32, Jan. 2022.

\bibitem{Wu2021}
Y.~Wu, H.-N. Dai, and H.~Wang, ``Convergence of blockchain and edge computing
  for secure and scalable {IIoT} critical infrastructures in industry 4.0,''
  \emph{IEEE Internet Things J.}, vol.~8, no.~4, pp. 2300--2317, Feb. 2021.

\bibitem{Xie2021}
L.~Xie, H.~T. Luan, Z.~Su, Q.~Xu, and N.~Chen, ``A game theoretical approach
  for secure crowdsourcing-based indoor navigation system with reputation
  mechanism,'' \emph{Early Access in IEEE Internet Things J., DOI:
  10.1109/JIOT.2021.3111999}.

\bibitem{Dai2021}
M.~Dai, Z.~Su, Q.~Xu, Y.~Wang, and N.~Lu, ``A trust-driven contract incentive
  scheme for mobile crowd-sensing networks,'' \emph{IEEE Trans. Veh. Technol.},
  vol.~71, no.~2, pp. 1794--1806, Feb. 2022.

\bibitem{Mu2020}
S.~Mu, Z.~Zhong, and D.~Zhao, ``Energy-efficient and delay-fair mobile
  computation offloading,'' \emph{IEEE Trans. Veh. Technol.}, vol.~69, no.~12,
  pp. 15\,746--15\,759, Dec. 2020.

\bibitem{Han2021iot}
H.~Hu, Q.~Wang, R.~Q. Hu, and H.~Zhu, ``Mobility-aware offloading and resource
  allocation in a {MEC}-enabled {IoT} network with energy harvesting,''
  \emph{IEEE Internet Things J.}, vol.~8, no.~24, pp. 17\,541--17\,556, Dec.
  2021.

\bibitem{Dai2018}
Y.~Dai, D.~Xu, S.~Maharjan, and Y.~Zhang, ``Joint computation offloading and
  user association in multi-task mobile edge computing,'' \emph{IEEE Trans.
  Veh. Technol.}, vol.~67, no.~12, pp. 12\,313--12\,325, Dec. 2018.

\bibitem{Jiang2020}
Q.~Jiang, N.~Zhang, J.~Ni, J.~Ma, X.~Ma, and K.-K.~R. Choo, ``Unified biometric
  privacy preserving three-factor authentication and key agreement for
  cloud-assisted autonomous vehicles,'' \emph{IEEE Tran. Veh. Technol.},
  vol.~69, no.~9, pp. 9390--9401, Sep. 2020.

\bibitem{Li2021}
J.~Li, Z.~Su, D.~Guo, K.-K.~R. Choo, and Y.~Ji, ``{PSL-MAAKA}: {Provably}
  secure and lightweight mutual authentication and key agreement protocol for
  fully public channels in internet of medical things,'' \emph{IEEE Internet
  Things J.}, vol.~8, no.~17, pp. 13\,183--13\,195, Sep. 2021.

\bibitem{Liu2020}
Y.~Liu, H.-H. Chen, L.~Wang, and W.~Meng, ``Artificial noisy {MIMO} systems
  under correlated scattering {Rayleigh} fading — {A} physical layer security
  approach,'' \emph{IEEE Syst. J.}, vol.~14, no.~2, pp. 2121--2132, Jun. 2020.

\bibitem{Yu2021}
Y.~Yu, S.~Liu, P.~L. Yeoh, B.~Vucetic, and Y.~Li, ``{LayerChain:} {A}
  hierarchical edge-cloud blockchain for large-scale low-delay industrial
  internet of things applications,'' \emph{IEEE Trans. Ind. Inf.}, vol.~17,
  no.~7, pp. 5077--5086, Jul. 2021.

\bibitem{Pyoung2020}
C.~K. Pyoung and S.~J. Baek, ``Blockchain of finite-lifetime blocks with
  applications to edge-based {IoT},'' \emph{IEEE Internet Things J.}, vol.~7,
  no.~3, pp. 2102--2116, Mar. 2020.

\bibitem{Pan2019}
J.~Pan, J.~Wang, A.~Hester, I.~Alqerm, Y.~Liu, and Y.~Zhao, ``{EdgeChain}: {An}
  edge-{IoT} framework and prototype based on blockchain and smart contracts,''
  \emph{IEEE Internet Things J.}, vol.~6, no.~3, pp. 4719--4732, Oct. 2019.

\bibitem{Xublock2020}
X.~Xu, X.~Zhang, H.~Gao, Y.~Xue, L.~Qi, and W.~Dou, ``{BeCome}:
  {B}lockchain-enabled computation offloading for {IoT} in mobile edge
  computing,'' \emph{IEEE Trans. Ind. Inf.}, vol.~16, no.~6, pp. 4187--4195,
  Jun. 2020.

\bibitem{Liublock2018}
M.~Liu, F.~R. Yu, Y.~Teng, V.~C.~M. Leung, and M.~Song, ``Computation
  offloading and content caching in wireless blockchain networks with mobile
  edge computing,'' \emph{IEEE Trans. Veh. Technol.}, vol.~67, no.~11, pp.
  11\,008--11\,021, Nov. 2018.

\bibitem{Chenblock2019}
W.~Chen, Z.~Zhang, Z.~Hong, C.~Chen, J.~Wu, S.~Maharjan, Z.~Zheng, and
  Y.~Zhang, ``Cooperative and distributed computation offloading for
  blockchain-empowered industrial {Internet of Things},'' \emph{IEEE Internet
  Things J.}, vol.~6, no.~5, pp. 8433--8446, Oct. 2019.

\bibitem{Fengblock2020}
J.~Feng, F.~Richard~Yu, Q.~Pei, X.~Chu, J.~Du, and L.~Zhu, ``Cooperative
  computation offloading and resource allocation for blockchain-enabled
  mobile-edge computing: {A} deep reinforcement learning approach,'' \emph{IEEE
  Internet Things J.}, vol.~7, no.~7, pp. 6214--6228, Jul. 2020.

\bibitem{Nguyen2021}
D.~C.~Nguyen, P.~N.~Pathirana, M.~Ding, and A.~Seneviratne, ``Secure
  computation offloading in blockchain based {IoT} networks with deep
  reinforcement learning,'' \emph{IEEE Trans. Network Sci. Eng.}, vol.~8,
  no.~4, pp. 3192--3208, Oct. 2021.

\bibitem{Fan2021}
Y.~Fan, L.~Wang, W.~Wu, and D.~Du, ``Cloud/edge computing resource allocation
  and pricing for mobile blockchain: {An} iterative greedy and search
  approach,'' \emph{IEEE Trans. Comput. Social Syst.}, vol.~8, no.~2, pp.
  451--463, Apr. 2021.

\bibitem{Seng2021}
S.~Seng, C.~Luo, X.~Li, H.~Zhang, and H.~Ji, ``User matching on blockchain for
  computation offloading in ultra-dense wireless networks,'' \emph{IEEE Trans.
  Network Sci. Eng.}, vol.~8, no.~2, pp. 1167--1177, Apr. 2021.

\bibitem{Sai2021}
S.~Xu, J.~Liu, and Y.~Cao, ``Intelligent reflecting surface empowered
  physical-layer security: {S}ignal cancellation or jamming?'' \emph{IEEE
  Internet Things J.}, vol.~9, no.~2, pp. 1265--1275, Jan. 2022.

\bibitem{Zheng2021}
Z.~Chu, W.~Hao, P.~Xiao, D.~Mi, Z.~Liu, M.~Khalily, J.~R. Kelly, and A.~P.
  Feresidis, ``Secrecy rate optimization for intelligent reflecting surface
  assisted {MIMO} system,'' \emph{IEEE Trans. Inf. Forensics Secur.}, vol.~16,
  pp. 1655--1669, 2021.

\bibitem{Yang2020}
L.~Yang, J.~Yang, W.~Xie, M.~O. Hasna, T.~Tsiftsis, and M.~D. Renzo, ``Secrecy
  performance analysis of {RIS}-aided wireless communication systems,''
  \emph{IEEE Trans. Veh. Technol.}, vol.~69, no.~10, pp. 12\,296--12\,300, Oct.
  2020.

\bibitem{Trigui2021}
I.~Trigui, W.~Ajib, and W.-P. Zhu, ``Secrecy outage probability and average
  rate of {RIS}-aided communications using quantized phases,'' \emph{IEEE
  Commun. Lett.}, vol.~25, no.~6, pp. 1820--1824, Jun. 2021.

\bibitem{FengIRS2020}
B.~Feng, Y.~Wu, M.~Zheng, X.-G. Xia, Y.~Wang, and C.~Xiao, ``Large intelligent
  surface aided physical layer security transmission,'' \emph{IEEE Trans.
  Signal Process.}, vol.~68, pp. 5276--5291, Sep. 2020.

\bibitem{Wu2021arch}
Y.~Wu, ``Cloud-edge orchestration for the {Internet of Things}: {Architecture}
  and {AI}-powered data processing,'' \emph{IEEE Internet Things J.}, vol.~8,
  no.~16, pp. 12\,792--12\,805, Aug. 2021.

\bibitem{Wu2021netw}
Y.~Wu, H.-N. Dai, H.~Wang, and K.-K.~R. Choo, ``Blockchain-based privacy
  preservation for {5G}-enabled drone communications,'' \emph{IEEE Netw.},
  vol.~35, no.~1, pp. 50--56, Jan. 2021.

\bibitem{Liugit1}
\BIBentryALTinterwordspacing
Y.~Liu, Z.~Su, and B.~Yu. Endogenous security of computation offloading in
  blockchain-empowered {Internet of Things}. [Online]. Available:
  \url{https://arxiv.org/abs/2203.01621}
\BIBentrySTDinterwordspacing

\bibitem{Wu2019}
Q.~Wu and R.~Zhang, ``Intelligent reflecting surface enhanced wireless network
  via joint active and passive beamforming,'' \emph{IEEE Trans. Wireless
  Commun.}, vol.~18, no.~11, pp. 5394--5409, Nov. 2019.

\bibitem{Yuiccc2019}
X.~Yu, D.~Xu, and R.~Schober, ``{MISO} wireless communication systems via
  intelligent reflecting surfaces : (invited paper),'' in \emph{2019 IEEE/CIC
  International Conference on Communications in China (ICCC)}, 2019, pp.
  735--740.

\bibitem{manopt}
N.~Boumal, B.~Mishra, P.-A. Absil, and R.~Sepulchre, ``{M}anopt, a {M}atlab
  toolbox for optimization on manifolds,'' \emph{Journal of Machine Learning
  Research}, vol.~15, no.~42, pp. 1455--1459, Apr. 2014.

\bibitem{absil2009optimization}
P.-A. Absil, R.~Mahony, and R.~Sepulchre, \emph{Optimization algorithms on
  matrix manifolds}.\hskip 1em plus 0.5em minus 0.4em\relax Princeton, New
  Jersey, US: Princeton University Press, 2009.

\bibitem{shewchuk1994introduction}
J.~R. Shewchuk \emph{et~al.}, \emph{An introduction to the conjugate gradient
  method without the agonizing pain}.\hskip 1em plus 0.5em minus 0.4em\relax
  Pittsburgh, PA, USA: Carnegie-Mellon University, Department of Computer
  Science, 1994.

\bibitem{VanChien2021}
T.~Van~Chien, L.~T. Tu, S.~Chatzinotas, and B.~Ottersten, ``Coverage
  probability and ergodic capacity of intelligent reflecting surface-enhanced
  communication systems,'' \emph{IEEE Commun. Lett.}, vol.~25, no.~1, pp.
  69--73, Jan. 2021.

\bibitem{weisstein2002crc}
E.~W. Weisstein, \emph{CRC concise encyclopedia of mathematics}.\hskip 1em plus
  0.5em minus 0.4em\relax Boca Raton, Florida, USA: CRC press, 2002.

\bibitem{duan2012scaling}
R.~Duan and H.-H. Su, ``A scaling algorithm for maximum weight matching in
  bipartite graphs,'' in \emph{Proceedings of the twenty-third annual ACM-SIAM
  symposium on Discrete Algorithms}.\hskip 1em plus 0.5em minus 0.4em\relax
  SIAM, 2012, pp. 1413--1424.

\bibitem{Wyner75}
A.~D. Wyner, ``The wire-tap channel,'' \emph{Bell System Technical J.},
  vol.~54, no.~8, pp. 1355--1367, 1975.

\bibitem{Liugit2}
\BIBentryALTinterwordspacing
Y.~Liu and Z.~Su, ``Minimization of secrecy outage probability in intelligent
  reflecting surfaces-assisted {MIMOME} system.'' [Online]. Available:
  \url{https://github.com/yiliangliu1990/liugit_pub/blob/master/IRS/IRS_PLS.pdf}
\BIBentrySTDinterwordspacing

\bibitem{Dinc2016}
T.~{Dinc}, A.~{Chakrabarti}, and H.~{Krishnaswamy}, ``A 60 {GHz CMOS}
  full-duplex transceiver and link with polarization-based antenna and {RF}
  cancellation,'' \emph{IEEE J. Solid-State Circuits}, vol.~51, no.~5, pp.
  1125--1140, May 2016.

\bibitem{Patnaik1949}
P.~Patnaik, ``The non-central $\chi$ 2-and f-distribution and their
  applications,'' \emph{Biometrika}, vol.~36, no. 1/2, pp. 202--232, Jun. 1949.

\bibitem{Kumar2015}
S.~Kumar, ``Exact evaluations of some {Meijer G}-functions and probability of
  all eigenvalues real for the product of two {Gaussian} matrices,''
  \emph{Journal of Physics A: Mathematical and Theoretical}, vol.~48, no.~44,
  p. 445206, Oct 2015.

\bibitem{Jeffrey2007}
A.~Jeffrey and D.~Zwillinger, \emph{Table of integrals, series, and
  products}.\hskip 1em plus 0.5em minus 0.4em\relax Burlington, MA, USA:
  Elsevier, 2007.

\end{thebibliography}

\end{document}